\documentclass[reqno, eucal]{amsart}

\usepackage[mathscr]{eucal} 
\usepackage[dvips]{graphicx}
\usepackage[dvips]{color}
\usepackage{amsmath,amsfonts,amssymb,amsthm,amscd}
\usepackage{epic}
\usepackage{eepic}
\usepackage{longtable}
\usepackage{array}

\newtheorem{theorem}{Theorem}
\newtheorem{lemma}[theorem]{Lemma}

\newtheorem{proposition}[theorem]{Proposition}
\newtheorem{corollary}[theorem]{Corollary}

\theoremstyle{definition}

\newtheorem{remark}[theorem]{Remark}

\newcommand{\Z}{{\mathbb Z}}
\newcommand{\Q}{{\mathbb Q}}

\newcommand{\C}{{\mathbb C}}

\newcommand{\I}{{\mathrm i}}

\newcommand{\ok}{{\rm{\bf k}}}
\newcommand{\OK}{{\rm{\bf K}}}
\newcommand{\am}{{\rm{\bf a}}^{\!-} }
\newcommand{\ap}{{\rm{\bf a}}^{\!+} }
\newcommand{\apm}{{\rm{\bf a}}^{\!\pm} }
\newcommand{\amp}{{\rm{\bf a}}^{\!\mp} }
\newcommand{\Am}{{\rm{\bf A}}^{\!-} }
\newcommand{\Ap}{{\rm{\bf A}}^{\!+} }
\newcommand{\Apm}{{\rm{\bf A}}^{\!\pm} }
\newcommand{\Amp}{{\rm{\bf A}}^{\!\mp} }
\newcommand{\h}{{\bf h}}

\newcommand{\hf}{{\scriptstyle \frac{1}{2}}}

\newcommand{\ot}{\otimes}

\begin{document}

\title[Reflection $K$ matrices]{Reflection $\boldsymbol{K}$ matrices associated with 
\\ an Onsager coideal of 
$\boldsymbol{U_p(A^{(1)}_{n-1})}, 
\boldsymbol{U_p(B^{(1)}_{n})}, 
\boldsymbol{U_p(D^{(1)}_{n})}$ 
and $\boldsymbol{U_p(D^{(2)}_{n+1})}$}

\author{Atsuo Kuniba}
\address{Atsuo Kuniba, Institute of Physics, 
University of Tokyo, Komaba, Tokyo 153-8902, Japan}
\email{atsuo.s.kuniba@gmail.com}

\author{Masato Okado}
\address{Masato Okado, Department of Mathematics, Osaka City University, 
Osaka, 558-8585, Japan}
\email{okado@sci.osaka-cu.ac.jp}

\author{Akihito Yoneyama}
\address{Akihito Yoneyama, Institute of Physics, 
University of Tokyo, Komaba, Tokyo 153-8902, Japan}
\email{yoneyama@gokutan.c.u-tokyo.ac.jp}

\maketitle

\vspace{0.5cm}
\begin{center}{\bf Abstract}
\end{center}

We determine the intertwiners of a family of Onsager coideal subalgebras 
of the quantum affine algebra $U_p(A^{(1)}_{n-1})$ 
in the fundamental representations and
$U_p(B^{(1)}_{n}), U_p(D^{(1)}_{n}), U_p(D^{(2)}_{n+1})$
in the spin representations.
They reproduce the reflection $K$ matrices 
obtained recently by the matrix product construction connected to the 
three dimensional integrability.
In particular the present approach provides the first proof of the 
reflection equation for the non type $A$ cases.

\vspace{0.4cm}

\section{Introduction}\label{s:int}

Reflection equation or boundary Yang-Baxter equation \cite{Ch,Kul,Sk} 
plays a fundamental role in quantum integrable systems with boundaries. 
It is a quartic relation of the so called $R$ and $K$ matrices
encoding the interactions in the bulk and at the boundary, respectively.
In the recent work \cite{KP} 
new solutions to the reflection equation 
including a spectral parameter
were obtained associated with the quantum affine algebra 
$U_p(A^{(1)}_{n-1})$ and 
conjectured for 
$U_p(D^{(2)}_{n+1}),
U_p(B^{(1)}_n)$ and $U_p(D^{(1)}_n)$.
They possess a number of distinct features.
Firstly the $K$ matrices act on the vector space of the form $(\C^2)^{\otimes n}$,
and as this structure may indicate, 
the companion $R$ matrices are those associated with the 
fundamental representations\footnote{The antisymmetric tensor 
representations of general degree between $0$ and $n$ whose degree 1 case 
corresponds to the vector representation. } of $U_p(A^{(1)}_{n-1})$
or the spin representation of 
$U_p(B^{(1)}_{n}), U_p(D^{(1)}_{n}), U_p(D^{(2)}_{n+1})$.
This is in contrast to the most preceding works which are concerned with 
the vector representation or $U_p(A^{(1)}_1)$.
See for example \cite{BFKZ, ML, MLu, NR,RV}.
Secondly the $K$ matrices are trigonometric 
and {\em dense} in that {\em all} the elements are 
nontrivial rational function of the (multiplicative) spectral parameter.
Thirdly each element of the $K$ matrices admits 
a {\em matrix product formula} 
in terms of $q$-Bosons.
Lastly the family of solutions labeled with $n$ are descendants of a {\em single} 
integrable structure in the three dimension called 
{\em quantized reflection equation} (Appendix \ref{app:qre}) 
which was extracted 
from the representation theory of 
the quantized coordinate ring $A_q(\mathrm{Sp}_4)$ 
\cite{KO1}\footnote{The relation between 
$q$ of $A_q$ and $p$ of $U_p$ is given by (\ref{qp}).}.

In this paper we show that the $K$ matrices in \cite{KP} are nothing 
but the intertwiners of the {\em Onsager coideal} subalgebras inside $U_p$,
and thereby establish the reflection equation 
in the non type $A$ cases 
$U_p(B^{(1)}_{n}), U_p(D^{(1)}_{n}), U_p(D^{(2)}_{n+1})$
for the first time.
The result gives a {\em characterization} of the $K$ matrices
in the framework of $U_p$, which is an essential complement 
to the {\em construction} by the matrix product method originating in 
$A_q(\mathrm{Sp}_4)$.

Let us briefly recall the approach to 
the reflection equation based on the coideal subalgebras of $U_p$,
which goes back to the affine Toda field theory with boundaries \cite{DM}.   
See also \cite{DM,Ko,RV} and references therein.
Let $e_i, f_i, k^{\pm 1}_i\,(i=0,\ldots,n')$ be the generators of 
the Drinfeld-Jimbo quantum affine algebra $U_p$ \cite{D,Ji}, where 
$n'$ is defined in (\ref{np}). 
The elements 
\begin{align}\label{ta}
b_i =  f_i + \beta_i  k^{-1}_i e_i + \gamma_i k^{-1}_i  \in U_p
\quad (0 \le i \le n')
\end{align}
behave as
$\Delta b_i = k_i^{-1}\otimes b_i +
(f_i + \beta_i k^{-1}_i e_i) \otimes 1$
under the coproduct $\Delta$ in (\ref{Del})
for any coefficients $\beta_i, \gamma_i$.
Thus the subalgebra $\mathcal{B}\subset U_p$ 
generated by $b_0,\ldots, b_{n'}$ satisfies 
$\Delta \mathcal{B} \subset U_p \otimes \mathcal{B}$
meaning that it is a left coideal\footnote{
Choosing the left and the right coideals is a matter of convention
transformed by an automorphism of $U_p$. They fit 
the reflection at the right and the left boundary of the system, respectively.
Here we adopt the left coideal to adjust to \cite{KP}.}.
Regarding the boundary reflection
as a linear map $K$ from a representation $(\pi_z,V_z)$ of $U_p$ to 
its dual $(\pi^\ast_{z^{-1}}, V^\ast_{z^{-1}})$ with the reversed spectral parameter, 
one has 
\begin{align}\label{tdk}
K(z,p): V_z \rightarrow V^\ast_{z^{-1}},
\qquad
K(z,p)\pi_z(b) = \pi^\ast_{z^{-1}}(b)K(z,p)\quad (b \in \mathcal{B}).
\end{align}
Here the latter equation is postulating that the $U_p$-symmetry in the bulk 
is not preserved fully at the boundary 
but still survives as a smaller symmetry corresponding to $\mathcal{B}$.
By composing the above $K$ and the $R$ matrices, 
one can construct a linear map
$V_x \otimes V_y \rightarrow V^\ast_{x^{-1}} \otimes V^\ast_{y^{-1}}$ 
commuting with $\Delta\mathcal{B}$ in two ways
along the two sides of the reflection equation.
See (\ref{zu}).
Therefore if a solution $K$ to (\ref{tdk}) exists uniquely and 
$V_x \otimes V_y$ is irreducible as a $\mathcal{B}$ module,
the reflection equation follows as a corollary.
In other words, the intertwining relation (\ref{tdk}) for such a coideal $\mathcal{B}$
achieves {\em linearization} of the reflection equation,
eliminating the task of proving the original quartic relation ``manually".
This is a boundary analogue of the classic idea that 
the cubic Yang-Baxter equation is attributed to the linear equation 
$[\check{R}, \Delta U_p]$ =0 representing the $U_p$-symmetry \cite{D,Ji}. 

The coideal $\mathcal{B}$ must be small enough;
otherwise the intertwining relation (\ref{tdk}) may not allow a solution.
Nonetheless it must also be large enough;
otherwise $V_x \otimes V_y$ may not become irreducible 
as a $\mathcal{B}$ module.
In this way one is led to a fundamental question;
what is the right ``size" or choice of the coideal $\mathcal{B}$  
in order to make the linearization work legitimately for a
given representation $V_z$?
To our knowledge it is still an outstanding issue in general.

This paper presents such coideals
$\mathcal{B}=\langle b_0, \ldots, b_{n'}\rangle$ that control the $K$ matrices
for the prototypical setting;  
the fundamental representations of $U_p(A^{(1)}_{n-1})$
and the spin representation of 
$U_p(B^{(1)}_{n})$, $U_p(D^{(1)}_{n})$, $U_p(D^{(2)}_{n+1})$.
They are realized by a specific choice of the coefficients in (\ref{ta}).
It turns out in all the cases that the resulting 
generators $b_i$'s form a closed set of relations among themselves
known as the {\em generalized $p$-Onsager algebras} \cite{BB}.
This algebra has drawn considerable attention since 
its first appearance at $p=1$ \cite{On} 
and has been generalized significantly.
See for example \cite[Rem. 9.1]{T} for the early history, 
\cite[Sec.1(1)]{Ko} for an account of more recent studies 
and the references therein. 

In this paper a left coideal subalgebra of $U_p$ 
isomorphic to a generalized $p$-Onsager algebra 
will just be called an {\em Onsager coideal} for short.
In this terminology all the sought coideals in our examples are Onsager coideals.
This is another main observation in this paper.
Although it is yet to be understood conceptually, we remark that
a quite parallel result has already been established in \cite{KOY}
for another prominent example; 
the symmetric tensor representations of $U_p(A^{(1)}_{n-1})$.

The paper is organized as follows.
In Section \ref{sec:qa} 
basic definitions of 
the relevant quantum affine algebras and 
their representations are recalled.
Then the specific Onsager coideals are introduced.
Their intertwiners are unique and  
automatically satisfy the reflection equation. 
In Section \ref{sec:mp} 
matrix product construction of the 
$R$ matrices \cite{BS,KS} and $K$ matrices \cite{KP} 
are quoted in the form adapted to the present setting.
In Section \ref{sec:mr}
the main result of the paper, 
Theorem \ref{th:A}, \ref{th:BD2} and \ref{th:D} are stated.
The latter two establish for the first time  
the matrix product solutions in the non type $A$ cases.
In Section \ref{sec:p} 
a proof of the main result is given.
In Section \ref{s:gen}
a multi-parameter generalization is formulated and 
the relation to the Onsager coideals is explained.
Section \ref{s:cdk} is a summary.

Appendix \ref{app:qre} 
recalls the quantized reflection equation \cite{KP}
which is the basic ingredient 
in the three dimensional approach.
This description is included since Section \ref{s:gen} requires 
a slight parametric generalization of \cite{KP}.
Appendix \ref{app:B} 
lists the explicit forms of the generators of the 
Onsager coideals.
We use the notation:
\begin{align*}
(z;q)_m = \prod_{j=1}^m(1-zq^{j-1}),
\qquad 
\binom{m}{r}_q = \frac{(q;q)_m}{(q;q)_{m-r}(q;q)_r}.
\end{align*}

\section{Quantum affine algebras}\label{sec:qa}

Let 
$U_p=U_p(A^{(1)}_{n-1})\,(n\ge 2), 
U_p(D^{(2)}_{n+1})\, (n \ge 2), 
U_p(B^{(1)}_n)\,(n\ge 3),
U_p(\tilde{B}^{(1)}_n)\,(n\ge 3), 
U_p(D^{(1)}_n)\, (n\ge 3)$
be the quantum affine algebras without derivation operator \cite{D,Ji}. 
We assume that $p$ is generic throughout.
For convenience we use the notation
\begin{align}
n'&=\begin{cases}
n-1 & \text{for}\; \, A^{(1)}_{n-1},\\
n  & \text{for } \;D^{(2)}_{n+1}, B^{(1)}_n, \tilde{B}^{(1)}_n, D^{(1)}_n,
\end{cases}
\label{np}
\\
\mathfrak{g}^{\mathrm{tr}} &= A^{(1)}_{n-1},\quad
\mathfrak{g}^{1,1} = D^{(2)}_{n+1},\quad
\mathfrak{g}^{2,1} = B^{(1)}_{n},\quad
\mathfrak{g}^{1,2} = \tilde{B}^{(1)}_{n},\quad
\mathfrak{g}^{2,2} = D^{(1)}_n.
\end{align}
$U_p$ is a Hopf algebra 
generated by $e_i, f_i, k^{\pm 1}_i\, (0 \le i \le n')$ satisfying the relations
\begin{equation}\label{uqdef}
\begin{split}
&k_i k^{-1}_i = k^{-1}_i k_i = 1,\quad [k_i, k_j]=0,\\
&k_ie_jk^{-1}_i = p_i^{a_{ij}}e_j,\quad 
k_if_jk^{-1}_i = p_i^{-a_{ij}}f_j,\quad
[e_i, f_j]=\delta_{i,j}\frac{k_i-k^{-1}_i}{p_i-p^{-1}_i},\\
&\sum_{\nu=0}^{1-a_{ij}}(-1)^\nu
e^{(1-a_{ij}-\nu)}_i e_j e_i^{(\nu)}=0,
\quad
\sum_{\nu=0}^{1-a_{ij}}(-1)^\nu
f^{(1-a_{ij}-\nu)}_i f_j f_i^{(\nu)}=0\;\;(i\neq j),
\end{split}
\end{equation}
where $e^{(\nu)}_i = e^\nu_i/[\nu]_{p_i}!, \,
f^{(\nu)}_i = f^\nu_i/[\nu]_{p_i}!$
and 
$[m]_p! = \prod_{k=1}^m [k]_p$ with 
$[m]_p = \frac{p^m-p^{-m}}{p-p^{-1}}$.
The constants $p_i\, (0 \le i \le n')$ in (\ref{uqdef}) are all 
taken as $p_i=p^2$ except the following:
\begin{align}
p_0 = p_n = p \;\;\text{for }D^{(2)}_{n+1},\qquad
p_n = p \;\; \text{for }B^{(1)}_{n},\qquad
p_0 = p \;\; \text{for }\tilde{B}^{(1)}_{n}.
\end{align}
Note that $p_0=p^{r}, \,p_n=p^{r'}$ for $\mathfrak{g}^{r,r'}$.
The Cartan matrix $(a_{ij})_{0 \le i,j \le n'}$ is 
determined from the Dynkin diagrams 
of the relevant affine Lie algebras according to the convention of \cite{Kac}:

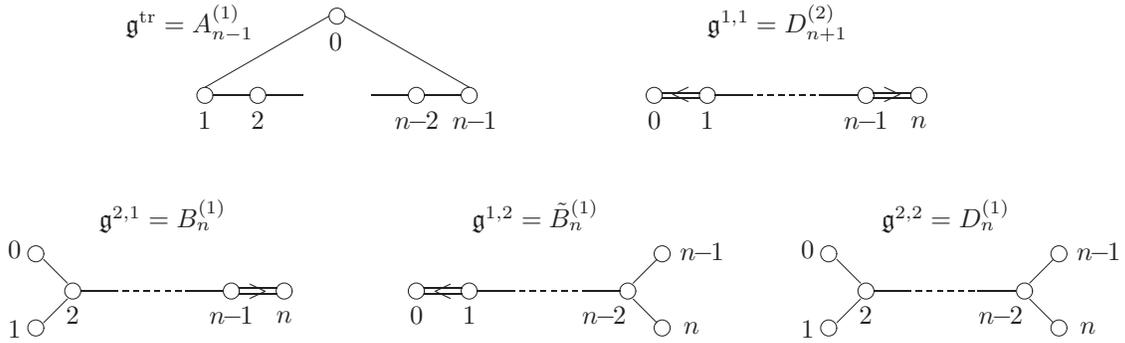
\begin{picture}(400, 150)(-31,-8)

\put(30,94){
\put(-10,24){$\mathfrak{g}^{\mathrm{tr}}= A^{(1)}_{n-1}$}
\drawline(20,3)(67,30)
\put(70,30){\circle{6}}
\drawline(73,30)(120,3)
\multiput( 20,0)(20,0){2}{\circle{6}}
\multiput(100,0)(20,0){2}{\circle{6}}
\multiput(23,0)(20,0){2}{\line(1,0){14}}
\put(83,0){\line(1,0){14}}\put(103,0){\line(1,0){14}}
\put(20,-6){\makebox(0,0)[t]{$1$}}
\put(40,-6){\makebox(0,0)[t]{$2$}}
\put(100,-6){\makebox(0,0)[t]{$n\!\! -\!\! 2$}}
\put(122,-6){\makebox(0,0)[t]{$n\!\! -\!\! 1$}}
\put(67,17){0}
}


\put(220,94){
\put(20,24){$\mathfrak{g}^{1,1}= D^{(2)}_{n+1}$}
\multiput( 0,0)(20,0){2}{\circle{6}}
\multiput(80,0)(20,0){2}{\circle{6}}
\put(23,0){\line(1,0){14}}
\put(62.5,0){\line(1,0){14}}
\multiput(2.85,-1)(0,2){2}{\line(1,0){14.3}} 
\multiput(82.85,-1)(0,2){2}{\line(1,0){14.3}} 
\multiput(39,0)(4,0){6}{\line(1,0){2}} 
\put(10,0.2){\makebox(0,0){$<$}}
\put(90,0.2){\makebox(0,0){$>$}}
\put(0,-6){\makebox(0,0)[t]{$0$}}
\put(20,-6){\makebox(0,0)[t]{$1$}}
\put(80,-6){\makebox(0,0)[t]{$n\!\! -\!\! 1$}}
\put(100,-7.8){\makebox(0,0)[t]{$n$}}
}

\put(-20,20){
\put(30,24){$\mathfrak{g}^{2,1}= B^{(1)}_{n}$}
\put(6,14){\circle{6}}\put(6,-14){\circle{6}}
\put(20,0){\circle{6}}
\multiput(80,0)(20,0){2}{\circle{6}}

\put(23,0){\line(1,0){14}}
\put(62.5,0){\line(1,0){14}}
\put(18,3){\line(-1,1){9}} \put(18,-3){\line(-1,-1){9}}

\multiput(82.85,-1)(0,2){2}{\line(1,0){14.3}} 
\multiput(39,0)(4,0){6}{\line(1,0){2}} 
\put(90,0){\makebox(0,0){$>$}}
\put(-2,19){\makebox(0,0)[t]{$0$}}
\put(-2,-11){\makebox(0,0)[t]{$1$}}
\put(20,-6){\makebox(0,0)[t]{$2$}}
\put(80,-6){\makebox(0,0)[t]{$n\!\! -\!\! 1$}}
\put(100,-7.8){\makebox(0,0)[t]{$n$}}
}


\put(130,20){
\put(21,24){$\mathfrak{g}^{1,2}=\tilde{B}^{(1)}_{n}$}
\put(93,14){\circle{6}}\put(93,-14){\circle{6}}
\multiput(0,0)(20,0){2}{\circle{6}}
\put(80,0){\circle{6}}
\put(23,0){\line(1,0){14}}
\put(63,0){\line(1,0){14}}

\put(82,3){\line(1,1){9}}\put(82,-3){\line(1,-1){9}}

\multiput(2.85,-1)(0,2){2}{\line(1,0){14.3}} 
\multiput(39,0)(4,0){6}{\line(1,0){2}} 
\put(10,0){\makebox(0,0){$<$}}
\put(108,18){\makebox(0,0)[t]{$n\!\! -\!\! 1$}}
\put(0,-6){\makebox(0,0)[t]{$0$}}
\put(20,-6){\makebox(0,0)[t]{$1$}}
\put(71,-6){\makebox(0,0)[t]{$n\!\! -\!\! 2$}}
\put(104,-12){\makebox(0,0)[t]{$n$}}
}


\put(280,20){
\put(26,24){$\mathfrak{g}^{2,2}= D^{(1)}_n$}
\put(6,14){\circle{6}}\put(6,-14){\circle{6}}
\put(20,0){\circle{6}}
\put(80,0){\circle{6}}
\put(93,14){\circle{6}}\put(93,-14){\circle{6}}

\put(18,3){\line(-1,1){9}} \put(18,-3){\line(-1,-1){9}}
\put(23,0){\line(1,0){14}}
\multiput(39,0)(4,0){6}{\line(1,0){2}} 
\put(62.5,0){\line(1,0){14}}
\put(82,3){\line(1,1){9}}\put(82,-3){\line(1,-1){9}}

\put(-2,19){\makebox(0,0)[t]{$0$}}
\put(-2,-11){\makebox(0,0)[t]{$1$}}
\put(20,-6){\makebox(0,0)[t]{$2$}}
\put(71,-6){\makebox(0,0)[t]{$n\!\! -\!\! 2$}}
\put(108,18){\makebox(0,0)[t]{$n\!\! -\!\! 1$}}
\put(104,-12){\makebox(0,0)[t]{$n$}}

}

\end{picture}

Here the affine Lie algebra $\tilde{B}^{(1)}_n$ is just 
$B^{(1)}_n$ but only with different enumeration of the nodes as shown above.
We keep it for uniformity of description.
Thus for instance in $U_p(D^{(2)}_{n+1})$, 
one has $a_{01}=-2, a_{10}=-1$ 
and $k_0e_1 = p^{-2} e_1 k_0$, $k_1 e_0 = p^{-2}e_0k_1$
and $k_1 e_1 = p^4 e_1k_1$.
We also note that $A^{(1)}_1$ is exceptional in that 
$a_{00}=a_{11}=-a_{01}=-a_{10}=2$.
We employ the coproduct $\Delta$ of the form 
\begin{align}\label{Del}
\Delta k^{\pm 1}_i = k^{\pm 1}_i\otimes k^{\pm 1}_i,\quad
\Delta e_i = 1\otimes e_i + e_i \otimes k_i,\quad
\Delta f_i = f_i\otimes 1 + k^{-1}_i\otimes f_i.
\end{align}
The opposite coproduct is denoted by $\Delta^{\!\mathrm{op}} = P \circ \Delta$,
where $P(u\otimes v) = v \otimes u$ is the exchange of the 
components.

\subsection{Representations $\pi_z, \pi^m_z, \pi^\pm_z$}\label{ss:rep}

We introduce the labeling set of the bases of the relevant representations as
\begin{align}
\mathrm{sp} &= 
\{\alpha=(\alpha_1,\ldots, \alpha_n) \in \{0,1\}^n\},
\label{sp}\\
\mathrm{sp}_m &= \{\alpha \in \mathrm{sp}\mid  |\alpha|=m\},
\quad |\alpha| = \alpha_1+\cdots + \alpha_n,
\label{spm}\\
\mathrm{sp}^\pm &=\{(\alpha_1,\ldots, \alpha_n) \in \mathrm{sp}
\mid (-1)^{|\alpha|}=\pm 1\}.
\end{align} 
Let 
${\bf e}_j = (0,\ldots, \overset{j\mathrm{th}}{1}, \ldots, 0) \in \Z^n$ be the 
elementary vector whose unique non vanishing element 1 is located 
at the $j$th component from the left. 
It should not be confused with the generator $e_j$ of $U_p$.
We consider the representation 
\begin{align}
\pi_z: \;U_p(\mathfrak{g}^{\mathrm{tr}}),  U_p(\mathfrak{g}^{r,r'})
\rightarrow \mathrm{End} V_z,
 \qquad V_z = \bigoplus_{\alpha \in \mathrm{sp}} 
\C(z) v_\alpha
\end{align}
given as follows:
\begin{alignat}{3}
\mathfrak{g}^{\mathrm{tr}} = A^{(1)}_{n-1};  \quad
e_jv_\alpha 
&= v_{\alpha-{\bf e}_j+{\bf e}_{j+1}},
&\quad
f_jv_\alpha
&= v_{\alpha+{\bf e}_j-{\bf e}_{j+1}},
&\quad
k_j v_\alpha
&= p^{2(\alpha_{j+1}-\alpha_{j})}v_{\alpha},\quad (j \in \Z_n),
\nonumber\\
\mathfrak{g}^{1,r'}= D^{(2)}_{n+1}, \tilde{B}^{(1)}_{n};   \quad
e_0v_\alpha 
&= zv_{\alpha+{\bf e}_1},
&\quad
f_0v_\alpha 
&= z^{-1}v_{\alpha-{\bf e}_1},
&\quad
k_0v_\alpha 
&= p^{2\alpha_1-1}v_{\alpha},
\nonumber\\
\mathfrak{g}^{2,r'}= B^{(1)}_{n}, D^{(1)}_{n};   \quad
e_0v_\alpha 
&= z^2v_{\alpha+{\bf e}_1+{\bf e}_2},
&\quad
f_0v_\alpha 
&= z^{-2}v_{\alpha-{\bf e}_1-{\bf e}_2},
&\quad
k_0v_\alpha 
&= p^{2(\alpha_1+\alpha_2-1)}v_{\alpha},
\nonumber\\
\mathfrak{g}^{r, r'}(1 \le r,r'\le 2);   \quad
e_jv_\alpha 
&= v_{\alpha-{\bf e}_j+{\bf e}_{j+1}},
&\quad
f_jv_\alpha
&= v_{\alpha+{\bf e}_j-{\bf e}_{j+1}},
&\quad
k_j v_\alpha
&= p^{2(\alpha_{j+1}-\alpha_{j})}v_{\alpha},\quad (0<j<n),
\nonumber\\
\mathfrak{g}^{r,1}=D^{(2)}_{n+1}, B^{(1)}_{n};   \quad
e_nv_\alpha 
&= v_{\alpha-{\bf e}_n},
&\quad
f_nv_\alpha 
&= v_{\alpha+{\bf e}_n},
&\quad
k_nv_\alpha 
&= p^{1-2\alpha_n}v_{\alpha},
\nonumber\\
\mathfrak{g}^{r,2}=\tilde{B}^{(1)}_{n}, D^{(1)}_{n};   \quad
e_nv_\alpha 
&= v_{\alpha-{\bf e}_{n-1}-{\bf e}_n},
&\quad
f_nv_\alpha 
&= v_{\alpha+{\bf e}_{n-1}+{\bf e}_n},
&\quad
k_nv_\alpha 
&= p^{2(1-\alpha_n-\alpha_{n-1})}v_{\alpha},
\label{rep}
\end{alignat}
where for example $f_j$ actually means $\pi_z(f_j)$.
The symbol $v_\beta$ on the RHS with $\beta \not\in \mathrm{sp}$
is to be understood as $0$. 
For $U_p(\mathfrak{g}^{\mathrm{tr}})$, the representation 
$(\pi_z, V_z)$ is decomposed as
\begin{align}\label{adec}
\pi^m_z: U_p(\mathfrak{g}^{\mathrm{tr}})  &\rightarrow V^m_z,
\qquad
V_z^m = \bigoplus_{\alpha \in \mathrm{sp}_m}\C(z)v_\alpha,
\qquad
V_z = V^0_z \oplus \cdots \oplus V^n_z,
\end{align}
where each component $(\pi^m_z, V^m_z)$ is irreducible and
called the $m$-th {\em fundamental representation}.
For $U_p(\mathfrak{g}^{r,r'})$, the representation
$(\pi_z, V_z)$ is irreducible except for 
$(r,r')=(2,2)$.
In the latter case, it decomposes into 
$(\pi^+_z, V^+_z)$ and $(\pi^-_z, V^-_z)$ as
\begin{align}\label{ddec}
\pi^\pm_z: U_p(\mathfrak{g}^{2,2})  &\rightarrow V^\pm_z,
\qquad
V^\pm_z = \bigoplus_{\alpha \in \mathrm{sp}^\pm}\C(z)v_\alpha,
\qquad
V_z = V^+_z \oplus V^-_z.
\end{align}
The representations $(\pi_z, V_z)$ and  
$(\pi^\pm_z, V^\pm_z)$ of $U_p(\mathfrak{g}^{r,r'})$ 
are called the {\em spin representations}.

\subsection{{\mathversion{bold}$\ast$}-dual representations}\label{ss:sd}

Let $\ast$ be the algebra anti-automorphism given by 
\begin{align}
e^\ast_i = e_i, \quad f^\ast_i = f_i, \quad k^\ast_i = k^{-1}_i.
\end{align}
The $\ast$-dual representation $(\pi^\ast_z,V^\ast_z)$
\begin{align}
\pi^\ast_z: \;U_p(\mathfrak{g}^{\mathrm{tr}}),  U_p(\mathfrak{g}^{r,r'})
\rightarrow \mathrm{End} V^\ast_z,
 \qquad V^\ast_z = \bigoplus_{\alpha \in \mathrm{sp}} 
\C(z) v^\ast_\alpha
\end{align}
of $(\pi_z,V)$ is defined by 
\begin{align}
\langle \pi^\ast_z(g)v^\ast, v'\rangle = 
\langle v^\ast, \pi_z(g^\ast)v'\rangle
\qquad 
(g \in U_p, v^\ast \in V^\ast_z, v' \in V_z).
\label{sd}
\end{align}
Here $\langle\;,\;\rangle$ denotes the dual pairing 
$\langle v^\ast_\alpha, v_\beta\rangle = \delta_{\alpha,\beta}$.
Practically in our case, the $\ast$-dual representations are obtained from 
(\ref{rep}) by formally identifying
$v^\ast_\alpha = v_{{\bf 1}-\alpha}$,
where ${\bf 1} = (1,\ldots, 1)= {\bf e}_1+\cdots + {\bf e}_n$.

According to (\ref{adec}) and (\ref{ddec}),
one has the decompositions:
\begin{align}
\pi^{m \ast}_z: U_p(\mathfrak{g}^{\mathrm{tr}})  
&\rightarrow V^{m \ast}_z,
\qquad
V^{m \ast}_z = \bigoplus_{\alpha \in \mathrm{sp}_m}\C(z)v^\ast_\alpha,
\qquad
V^\ast_z = V^{0\ast}_z \oplus \cdots \oplus V^{n\ast}_z,
\label{adec2}\\
\pi^{\pm \ast}_z: U_p(\mathfrak{g}^{2,2})  
&\rightarrow V^{\pm \ast}_z,
\qquad
V^{\pm \ast}_z 
= \bigoplus_{\alpha \in \mathrm{sp}^{\pm}}
\C(z)v^\ast_\alpha,
\qquad
V^\ast_z = V^{+\ast}_z \oplus V^{-\ast}_z.
\label{ddec2}
\end{align}

\subsection{Intertwiner of {\mathversion{bold}$U_p$}; 
{\mathversion{bold}$R$} matrix}
Intertwiners of $\Delta U_p$ are called $R$ matrices.
We will be concerned with the three kinds of $R$ matrices as 
\begin{alignat}{2}
\mathscr{R}(x/y,p) &\in \mathrm{End}(V_x\otimes V_y),
&\qquad \mathscr{R}(z,p) (v_\alpha \otimes v_\beta) &= 
\sum_{\gamma,\delta \in \mathrm{sp}} 
\mathscr{R}(z,p)^{\gamma,\delta}_{\alpha,\beta}
v_\gamma \otimes v_\delta,
\label{R}\\
\mathscr{R}^\ast(x/y,p) &\in \mathrm{End}(V_x\otimes V^\ast_y),
&\qquad \mathscr{R}^\ast(z,p) (v_\alpha \otimes v^\ast_\beta) &= 
\sum_{\gamma,\delta \in \mathrm{sp}}
\mathscr{R}^\ast(z,p)^{\gamma,\delta}_{\alpha,\beta}
v_\gamma \otimes v^\ast_\delta,
\label{Rs}\\
\mathscr{R}^{\ast\ast}(x/y,p) &\in \mathrm{End}(V^\ast_x\otimes V^\ast_y),
&\qquad \mathscr{R}^{\ast\ast}(z,p) (v^\ast_\alpha \otimes v^\ast_\beta) &= 
\sum_{\gamma,\delta \in \mathrm{sp}}
\mathscr{R}^{\ast\ast}(z,p)^{\gamma,\delta}_{\alpha,\beta}
v^\ast_\gamma \otimes v^\ast_\delta.
\label{Rss}
\end{alignat}
From the remark after (\ref{sd}), one may set
\begin{align}\label{rmm}
\mathscr{R}^\ast(z,p)^{\gamma,\delta}_{\alpha, \beta}
= \mathscr{R}(z,p)^{\gamma,{\bf 1}-\delta}_{\alpha,{\bf 1}-\beta},
\qquad
\mathscr{R}^{\ast\ast}(z,p)^{\gamma,\delta}_{\alpha, \beta}
= \mathscr{R}(z,p)^{{\bf 1}-\gamma,{\bf 1}-\delta}_{{\bf 1}-\alpha,{\bf 1}-\beta}.
\end{align}

The $R$ matrices are characterized up to normalization by the 
intertwining relations:
\begin{align}
\Delta^{\!\rm op}(g) \mathscr{R}(x/y,p) 
&= \mathscr{R}(x/y,p) \Delta(g), 
\label{ir2}\\
\Delta^{\!\rm op}(g) \mathscr{R}^\ast(x/y,p) 
&= \mathscr{R}^\ast(x/y,p) \Delta(g), 
\label{ir3}\\
\Delta^{\!\rm op}(g) \mathscr{R}^{\ast\ast}(x/y,p) 
&= \mathscr{R}^{\ast\ast}(x/y,p) \Delta(g)
\label{ir4}
\end{align}
for $g \in U_p$.
According to (\ref{adec}) and (\ref{ddec}), 
one has the decompositions:
\begin{alignat}{2}
\mathscr{R}(x/y,p) &= \bigoplus_{0 \le m,m' \le n}
\mathscr{R}_{m,m'}(x/y,p),
&\qquad
\mathscr{R}_{m,m'}(x/y,p) 
&\in \mathrm{End}(V^m_x\otimes V^{m'}_y)
\quad \text{for} \;\; U_p(\mathfrak{g}^{\mathrm{tr}}),
\label{rde1}\\
\mathscr{R}(x/y,p) &= \bigoplus_{\sigma,\sigma' =\pm}
\mathscr{R}_{\sigma,\sigma'}(x/y,p),
&\qquad
\mathscr{R}_{\sigma,\sigma'}(x/y,p) 
&\in \mathrm{End}(V^{\sigma}_x\otimes V^{\sigma'}_y)
\quad \text{for}\;\; U_p(\mathfrak{g}^{2,2})
\label{rde2}
\end{alignat}
and similarly for 
$\mathscr{R}^\ast(x/y,p)$ and $\mathscr{R}^{\ast\ast}(x/y,p)$.
The $R$ matrices satisfy the Yang-Baxter equations \cite{Bax,D,Ji}
\begin{align}
\mathscr{R}_{12}(z_{12},p)\mathscr{R}_{13}(z_{13},p)
\mathscr{R}_{23}(z_{23},p)
&= \mathscr{R}_{23}(z_{23},p)\mathscr{R}_{13}(z_{13},p)
\mathscr{R}_{12}(z_{12},p),
\\
\mathscr{R}_{12}(z_{12},p)\mathscr{R}^\ast_{13}(z_{13},p)
\mathscr{R}^\ast_{23}(z_{23},p)
&= \mathscr{R}^\ast_{23}(z_{23},p)\mathscr{R}^\ast_{13}(z_{13},p)
\mathscr{R}_{12}(z_{12},p),
\\
\mathscr{R}^\ast_{12}(z_{12},p)\mathscr{R}^\ast_{13}(z_{13},p)
\mathscr{R}^{\ast\ast}_{23}(z_{23},p)
&= \mathscr{R}^{\ast\ast}_{23}(z_{23},p)\mathscr{R}^\ast_{13}(z_{13},p)
\mathscr{R}^\ast_{12}(z_{12},p),
\\
\mathscr{R}^{\ast\ast}_{12}(z_{12},p)\mathscr{R}^{\ast\ast}_{13}(z_{13},p)
\mathscr{R}^{\ast\ast}_{23}(z_{23},p)
&= \mathscr{R}^{\ast\ast}_{23}(z_{23},p)\mathscr{R}^{\ast\ast}_{13}(z_{13},p)
\mathscr{R}^{\ast\ast}_{12}(z_{12},p),
\end{align}
where $z_{ij}=z_i/z_j$.
Here the notation is standard;
$\mathscr{R}_{ij}(z,p)$ for instance denotes the $\mathscr{R}(z,p)$ 
that acts on the $i$th and the $j$th components from the left 
in the three-fold tensor product and as the identity in the other one.
For $U_p(\mathfrak{g}^{\mathrm{tr}})$ and 
$U_p(\mathfrak{g}^{2,2})$, the Yang-Baxter equations actually 
hold in finer subspaces corresponding to the decompositions 
(\ref{rde1}) and (\ref{rde2}).

At this stage we do not specify the normalization of the $R$ matrices.
One typical choice will be given in (\ref{no1})--(\ref{no3}) via Theorem \ref{th:r}.
The basic properties of these $R$ matrices like the spectral decomposition
have been described in \cite{DO} for $U_p(A^{(1)}_{n-1})$,
in \cite{O} for $U_p(B^{(1)}_n)$, $U_p(D^{(1)}_n)$ and 
in \cite{KS} for $U_p(D^{(2)}_{n+1})$.

\subsection{Left coideal subalgebras}\label{ss:coi}

Let $\epsilon, s, s'$ be the parameters obeying
\begin{align}\label{ss12}
ss'=-\I \epsilon p^{-1}, \quad 
\epsilon = \pm 1,
\end{align}
where $\I = \sqrt{-1}$.
For $U_p(\mathfrak{g}^{\mathrm{tr}})$ we introduce 
the left coideal subalgebra $\mathcal{B}^{\mathrm{tr}}$
generated by $b_0, \ldots, b_{n-1}$ given by
\begin{align}
b_i &= f_i +p^{2} k^{-1}_ie_i 
+ \frac{\I \epsilon p}{1-p^2}k^{-1}_i\quad (i \in \Z_n).
\label{bai}
\end{align}
For $U_p(\mathfrak{g}^{r,r'})$ 
we introduce the left coideal subalgebra $\mathcal{B}^{r,r'}_{k,k'}$
labeled with $k,k' \in \{1,2\}$ satisfying $r\le k$ and $r'\le k'$.
So there are {\em nine} pairs
$\mathcal{B}^{r,r'}_{k,k'} \subset U_p(\mathfrak{g}^{r,r'})$ 
of such kind.
They are generated by $b_0, \ldots, b_n$ defined by
\begin{align}
b_0 &= f_0 + \Bigl(\frac{s'}{s}p\Bigr)^r k_0^{-1}e_0 
+d_{r,k}(s,s') k_0^{-1},
\label{b0}
\\
b_i &= f_i +p^{2} k^{-1}_ie_i 
+ \frac{\I \epsilon p}{1-p^2}k^{-1}_i\quad (0 < i <n),
\label{bi}\\
b_n &= f_n + \Bigl(\frac{s}{s'}p\Bigr)^{r'} k_n^{-1}e_n
+d_{r',k'}(s',s) k_n^{-1},
\label{bn}
\end{align}
where $d_{1,1}(u,v), d_{1,2}(u,v).d_{2,2}(u,v)$ are given by
\begin{align}\label{drk}
d_{1,1}(u,v) = (u^{-1}+v)\frac{\I p}{1-p^2},
\qquad
d_{1,2}(u,v) = 0,
\qquad
d_{2,2}(u,v) = \frac{\I \epsilon p v}{u(1-p^2)}.
\end{align}
For convenience we list the elements $b_i$  in Appendix \ref{app:B}.
The coideals $\mathcal{B}^{\mathrm{tr}}$ and $\mathcal{B}^{r,r'}_{k,k'}$ 
are examples of generalized $p$-Onsager algebras \cite{BB}.
This aspect will be explained in Section \ref{ss:o},
in a more generalized setting including further parameters.

\subsection{Intertwiner of coideal; {\mathversion{bold}$K$} matrix}\label{ss:cok}

Intertwiners of the coideals are called $K$ matrices.
It is a linear map 
\begin{align}\label{zuk}
\mathscr{K}(z,p):  V_z \rightarrow V^\ast_{z^{-1}},
\qquad
\mathscr{K}(z,p) v_\alpha = \sum_{\beta \in \mathrm{sp}} 
\mathscr{K}(z,p)^\beta_\alpha v^\ast_\beta
\end{align}
satisfying the intertwining relation
\begin{align}
\mathscr{K}(z,p) \pi_z(b) &= \pi^\ast_{z^{-1}}(b) \mathscr{K}(z,p),
\label{ir1}
\end{align}
where $b \in \mathcal{B}^{\mathrm{tr}}$ for $U_p(\mathfrak{g}^{\mathrm{tr}})$ 
and $b \in \mathcal{B}^{r,r'}_{k,k'}$ for $U_p(\mathfrak{g}^{r,r'})$.
If $V_z$ is irreducible as a module over 
$\mathcal{B}^{\mathrm{tr}}$ or  
$\mathcal{B}^{r,r'}_{k,k'}$,
then (\ref{ir1}) characterizes the $K$ matrix 
up to normalization.

Consider the two maps going from 
$V_x\otimes V_y$ to $V^\ast_{x^{-1}} \otimes V^\ast_{y^{-1}}$ 
composed of $R$ and $K$ matrices as
\begin{align}\label{zu}
\begin{picture}(200,210)(-70,50)
\put(0,250){$V_x \otimes V_y$}
\put(-5,245){\vector(-2,-1){45}}\put(-65,239){$P\mathscr{R}(x/y)$}
\put(38,245){\vector(2,-1){45}}\put(61,240){$\mathscr{K}_2(y)$}
\put(-70,207){$V_y \otimes V_x$}
\put(-52,197){\vector(0,-1){32}}\put(-83,181){$\mathscr{K}_2(x)$}
\put(-70,150){$V_y \otimes V^\ast_{x^{-1}}$}
\put(-52,141){\vector(0,-1){32}}\put(-96,123){$P\mathscr{R}^\ast(xy)$}
\put(-72,95){$V^\ast_{x^{-1}}\otimes V_y$}
\put(70,207){$V_x \otimes V^\ast_{y^{-1}}$}
\put(88,197){\vector(0,-1){32}}\put(93,181){$P\mathscr{R}^\ast(xy)$}
\put(67,150){$V^\ast_{y^{-1}}\otimes V_x$}
\put(88,141){\vector(0,-1){32}}\put(93,123){$\mathscr{K}_2(x)$}
\put(66,95){$V^\ast_{y^{-1}}\otimes V^\ast_{x^{-1}}$}
\put(-48,87){\vector(2,-1){40}}\put(-61,70){$\mathscr{K}_2(y)$}
\put(87,87){\vector(-2,-1){37}}\put(75,70){$P\mathscr{R}^{\ast\ast}(x/y)$}
\put(0,60){$V^\ast_{x^{-1}} \otimes V^\ast_{y^{-1}}$}
\end{picture}
\end{align}
Here $P$ is the transposition introduced after (\ref{Del}), 
$\mathscr{R}(z), \mathscr{R}^\ast(z), \mathscr{R}^{\ast\ast}(z)$ 
are abbreviation of 
$\mathscr{R}(z,p)$, 
$\mathscr{R}^\ast(z,p)$, 
$\mathscr{R}^{\ast\ast}(z,p)$, respectively and 
$\mathscr{K}_2(z) = 1 \otimes \mathscr{K}(z,p)$.
Since the $K$ matrices act only on the right component,
the both maps commute with 
$\Delta \mathcal{B} \subset U_p \otimes \mathcal{B}$
($\mathcal{B} = \mathcal{B}^{\mathrm{tr}}$ or $\mathcal{B}^{r,r'}_{k,k'}$)
owing to the intertwining relations 
(\ref{ir2})--(\ref{ir4}) and (\ref{ir1}).
Therefore if 
$V_x \otimes V_y$ with generic $x,y$ 
is irreducible as a $\mathcal{B}$ module and the $R$ and the $K$ matrices are 
properly normalized, 
the above diagram implies the reflection equation \cite{DM}
\begin{align}\label{re1}
\mathscr{K}_2(y)\mathscr{R}^{\ast}_{21}(xy)
\mathscr{K}_1(x)\mathscr{R}_{12}(x/y)
=
\mathscr{R}^{\ast\ast}_{21}(x/y)\mathscr{K}_1(x)
\mathscr{R}^{\ast}_{12}(xy)\mathscr{K}_2(y),
\end{align}
where 
$\mathscr{R}_{12}(z) = \mathscr{R}(z,p)$, 
$\mathscr{R}^\ast_{21}(z)= P\mathscr{R}^\ast(z,p)P$,
$\mathscr{R}^\ast_{12}(z)= \mathscr{R}^\ast(z,p)$,
$\mathscr{R}^{\ast\ast}_{21}(z)= P\mathscr{R}^{\ast\ast}(z,p)P$ 
and 
$\mathscr{K}_1(z) = \mathscr{K}(z,p) \otimes 1$.
In terms of the matrix elements for the 
transition $v_\alpha \otimes v_\beta \mapsto 
v^\ast_{\alpha'''} \otimes v^\ast_{\beta'''}$, it reads
\begin{equation}\label{re12}
\begin{split}
&\sum_{\alpha',\alpha'', \beta', \beta'' \in \mathrm{sp}} 
\mathscr{K}(y,p)^{\beta'''}_{\beta''}
\mathscr{R}^\ast(xy,p)^{\beta'', \alpha'''}_{\beta', \alpha''}
\mathscr{K}(x,p)^{\alpha''}_{\alpha'}
\mathscr{R}(x/y,p)^{\alpha', \beta'}_{\alpha, \beta}
\\
&=\sum_{\alpha',\alpha'', \beta', \beta'' \in \mathrm{sp}} 
\mathscr{R}^{\ast\ast}(x/y,p)^{\beta''', \alpha'''}_{\beta'',\alpha''}
\mathscr{K}(x,p)^{\alpha''}_{\alpha'}
\mathscr{R}^\ast(xy,p)^{\alpha', \beta''}_{\alpha, \beta'}
\mathscr{K}(y,p)^{\beta'}_\beta.
\end{split}
\end{equation}
As the set of equations this is equivalent, due to (\ref{rmm}), to 
\begin{equation}\label{re123}
\begin{split}
&\sum_{\alpha',\alpha'', \beta', \beta'' \in \mathrm{sp}} 
\mathscr{K}^\ast(y,p)^{\beta'''}_{\beta''}
\mathscr{R}(xy,p)^{\beta'', \alpha'''}_{\beta', \alpha''}
\mathscr{K}^\ast(x,p)^{\alpha''}_{\alpha'}
\mathscr{R}(x/y,p)^{\alpha', \beta'}_{\alpha, \beta}
\\
&=\sum_{\alpha',\alpha'', \beta', \beta'' \in \mathrm{sp}} 
\mathscr{R}(x/y,p)^{\beta''', \alpha'''}_{\beta'',\alpha''}
\mathscr{K}^\ast(x,p)^{\alpha''}_{\alpha'}
\mathscr{R}(xy,p)^{\alpha', \beta''}_{\alpha, \beta'}
\mathscr{K}^\ast(y,p)^{\beta'}_\beta,
\end{split}
\end{equation}
where $\mathscr{K}^\ast(z,p)^\beta_\alpha 
:= \mathscr{K}(z,p)^{{\bf 1}-\beta}_{\alpha}$.
This is the form which essentially agrees with 
\cite[eqs.(73), (82)]{KP}.
The $K(z)$ in \cite{KP} corresponds to 
$\mathscr{K}^\ast(z,p)$ here 
having the matrix elements $\mathscr{K}^\ast(z,p)^\beta_\alpha$.
As seen in (\ref{re123}), using $\mathscr{K}^\ast(z,p)$ 
enables one to formulate the reflection equation without 
involving $\mathscr{R}^\ast(z,p)$ and $\mathscr{R}^{\ast\ast}(z,p)$. 

When the intertwining relations
(\ref{ir2})--(\ref{ir4}) and (\ref{ir1}) 
are decomposed into the ones in finer irreducible submodules over the coideals,  
the same argument shows that the reflection equation holds individually 
in those subspaces.

At this stage we do not specify the normalization of the $K$ matrices.
One typical example will be given in (\ref{no4})--(\ref{no6})
via Theorem \ref{th:A}, \ref{th:BD2} and \ref{th:D}. 
The reflection equations (\ref{re12}) and (\ref{re123}) are
diagrammatically expressed as follows:

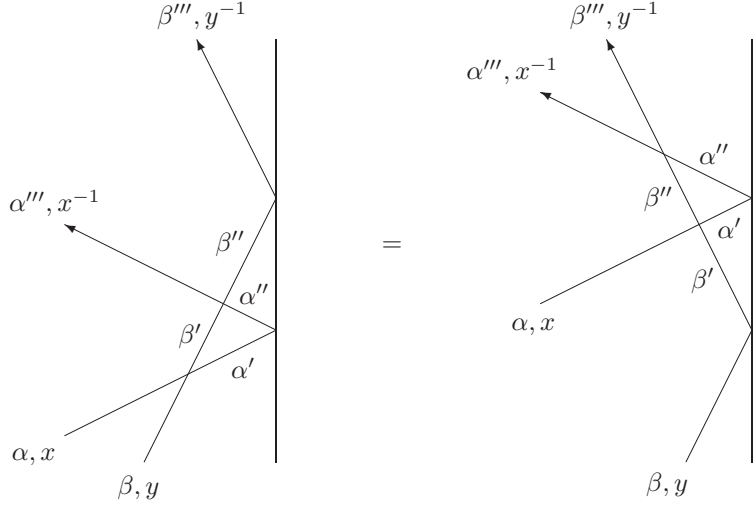
\begin{picture}(300,200)(-70,-18)


\put(100,0){
\put(0,0){\line(0,1){160}}

\put(0,100){\vector(-1,2){30}}\put(-45,166){$\beta''', y^{-1}$}
\put(0,100){\line(-1,-2){50}}\put(-60,-12){$\beta, y$}

\put(0,50){\vector(-2,1){80}}\put(-101,95){$\alpha''', x^{-1}$}
\put(0,50){\line(-2,-1){80}}\put(-100,1){$\alpha, x$}

\put(-23,80){$\beta''$}
\put(-37,45){$\beta'$}

\put(-14,60){$\alpha''$}
\put(-17,32){$\alpha'$}
}

\put(280,0){
\put(-140,80){$=$}
\put(0,0){\line(0,1){160}}

\put(0,100){\vector(-2,1){80}}\put(-69,166){$\beta''', y^{-1}$}
\put(0,100){\line(-2,-1){80}}\put(-91,50){$\alpha, x$}

\put(0,50){\vector(-1,2){55}}\put(-108,145){$\alpha''', x^{-1}$}
\put(0,50){\line(-1,-2){25}}\put(-40,-11){$\beta, y$}

\put(-20,113){$\alpha''$}
\put(-13,85){$\alpha'$}

\put(-23,66){$\beta'$}
\put(-41,97){$\beta''$}
}
\end{picture}

The boundary reflection changes 
the spectral parameters $x$ and $y$ into $x^{-1}$ and $y^{-1}$.

\section{Matrix product construction}\label{sec:mp}

\subsection{{\mathversion{bold}$q$}-Bosons}

Let 
$F_q = \bigoplus_{m\ge 0}\C|m\rangle$
be the Fock space equipped with 
the $q$-Boson operators $\apm, \ok$.
Similarly let $F_{q^2}$ be the one with 
$\Apm, \OK$ as follows:
\begin{alignat}{3}
\am|m\rangle &= (1-q^{2m})|m-1\rangle,
\quad & \ap|m\rangle &= |m+1\rangle,
\quad & \ok&=q^{\bf h},
\\
\Am|m\rangle &= (1-q^{4m})|m-1\rangle,
\quad & \Ap|m\rangle &= |m+1\rangle,
\quad & \OK &= q^{2{\bf h}},
\\
{\bf h}|m\rangle &= m|m\rangle.
\end{alignat}
We use the same notation for the base vectors either for 
$F_q$ or $F_{q^2}$, which will not cause a confusion.
The following $q$-Boson relations hold\footnote{The operators $\ok, \OK$
in \cite{KP} differ from the ones here by the factors $q^\hf, q$ 
corresponding to the zero point energy.}:
\begin{align}\label{qb}
\ok \,\apm &= q^{\pm 1} \apm \ok,
\quad \apm \,\amp = 1-q^{1\mp 1}\ok^2,
\quad
\OK \,\Apm = q^{\pm 2} \Apm \OK,
\quad \Apm \Amp = 1-q^{2\mp 2}\OK^2.
\end{align}

Denote the dual space of $F_{q^r}\,(r=1,2)$ by 
$F^\ast_{q^r}= \bigoplus_{m \ge 0}\C\langle m|$  
such that $\langle m |m'\rangle = \delta_{m,m'}(q^{2r};q^{2r})_m$.
We endow them with the $q$-Boson action by 
\begin{alignat}{2}
\langle m | \ap &= \langle m-1 |(1-q^{2m}) ,
\quad & \langle m | \am &= \langle m+1 |,
\\
\langle m | \Ap &= \langle m-1 |(1-q^{4m}),
\quad & \langle m | \Am &= \langle m+1 |,
\\
\langle m | {\bf h} &= \langle m | m.
\end{alignat}
These definitions satisfy 
$(\langle m |X)|m'\rangle = \langle m |(X|m'\rangle)$.

\subsection{Boundary vectors}

For $r=1,2$, we introduce the elements called {\em boundary vectors}:
\begin{align}
&\langle \chi_r| = \sum_{m\ge 0}\frac{\langle rm|}{(q^{2r^2};q^{2r^2})_m} 
\in F^\ast_{q^2},\qquad\,
|\chi_r\rangle = \sum_{m\ge 0}\frac{|rm\rangle}{(q^{2r^2}; q^{2r^2})_m}
\in F_{q^2},
\label{xk}\\
&\langle\eta_r| = \sum_{m\ge 0}\frac{\langle rm|}{(q^{r^2};q^{r^2})_m} \in F^\ast_q,
\qquad\quad
|\eta_r\rangle = \sum_{m\ge 0}\frac{|rm\rangle}{(q^{r^2};q^{r^2})_m} \in F_q.
\label{xb}
\end{align}
They will be utilized 
in the matrix product constructions (\ref{Rrr}) and (\ref{kkk}).
They are characterized by 
\begin{align}
|\chi_r\rangle &= |\eta_r\rangle|_{q\rightarrow q^2},
\quad\qquad\qquad\quad\;\;\,
\langle \chi_r| = \langle \eta_r||_{q\rightarrow q^2},
\label{ce}\\
\apm |\eta_1\rangle &= (1 \mp q^{\hf(1\mp 1)}\ok)|\eta_1\rangle,
\qquad\,
\langle \eta_1| \apm = \langle \eta_1| (1\pm q^{\hf(1\pm 1)}\ok),
\label{et1}\\
\ap | \eta_2\rangle &= \am  | \eta_2\rangle, \qquad\qquad\qquad
\quad\;
\langle \eta_2 | \ap = \langle \eta_2 | \am.
\label{et2}
\end{align} 

\subsection{Matrix product construction from {\mathversion{bold}$L$}}
Define an operator $L$  by \cite{BS, KP}
\begin{align}
L= \begin{pmatrix}
L_{0,0}^{0,0} & 
L_{0,1}^{0,0} & 
L_{1,0}^{0,0} & 
L_{1,1}^{0,0} \\
L_{0,0}^{0,1} & 
L_{0,1}^{0,1} & 
L_{1,0}^{0,1} & 
L_{1,1}^{0,1} \\
L_{0,0}^{1,0} & 
L_{0,1}^{1,0} & 
L_{1,0}^{1,0} & 
L_{1,1}^{1,0} \\
L_{0,0}^{1,1} & 
L_{0,1}^{1,1} & 
L_{1,0}^{1,1} & 
L_{1,1}^{1,1} 
\end{pmatrix} 
= \begin{pmatrix}
1 & 0 & 0 & 0 \\
0 & \I \epsilon q \OK & \Am & 0\\
0 & \Ap & \I \epsilon q \OK & 0\\
0 & 0 & 0 & 1
\end{pmatrix}\qquad
(\epsilon=\pm 1).
\label{L}
\end{align}
We introduce the linear operators
$R(z,q), R^\ast(z,q), R^{\ast\ast}(z,q)$ 
with $R= R^{\mathrm{tr}}, R^{r,r'}$ with $r,r' \in \{1,2\}$  by
\begin{alignat}{2}
R(x/y,q) &\in \mathrm{End}(V_x\otimes V_y),
&\qquad R(z,q) (v_\alpha \otimes v_\beta) &= 
\sum_{\gamma,\delta \in \mathrm{sp}} R(z,q)^{\gamma,\delta}_{\alpha,\beta}
v_\gamma \otimes v_\delta,
\label{R1}\\
R^\ast(x/y,q) &\in \mathrm{End}(V_x\otimes V^\ast_y),
&\qquad R^\ast(z,q) (v_\alpha \otimes v^\ast_\beta) &= 
\sum_{\gamma,\delta  \in \mathrm{sp}} R^\ast(z,q)^{\gamma,\delta}_{\alpha,\beta}
v_\gamma \otimes v^\ast_\delta,
\label{Rs1}\\
R^{\ast\ast}(x/y,q) &\in \mathrm{End}(V^\ast_x\otimes V^\ast_y),
&\qquad R^{\ast\ast}(z,q) (v^\ast_\alpha \otimes v^\ast_\beta) &= 
\sum_{\gamma,\delta  \in \mathrm{sp}}
R^{\ast\ast}(z,q)^{\gamma,\delta}_{\alpha,\beta}
v^\ast_\gamma \otimes v^\ast_\delta,
\label{Rss1}
\end{alignat}
where the elements are given by the matrix product formulas:
\begin{align}
R^{\mathrm{tr}}(z,q)_{\alpha,\beta}^{\gamma,\delta} &=  
\varrho^{\mathrm{tr}}_{|\alpha|, |\beta|}(z)
\mathrm{Tr}\bigl(z^{{\bf h}}
L^{\gamma_1,\delta_1}_{\alpha_1, \beta_1}
\cdots
L^{\gamma_n,\delta_n}_{\alpha_n, \beta_n}\bigl),
\label{Rtr}
\\
R^{r,r'}(z,q)_{\alpha,\beta}^{\gamma,\delta}
 &=  \varrho^{r,r'}(z)
\langle \chi_r |z^{{\bf h}}
L^{\gamma_1,\delta_1}_{\alpha_1, \beta_1}
\cdots
L^{\gamma_n,\delta_n}_{\alpha_n, \beta_n}|\chi_{r'}\rangle\quad 
((r,r') \neq (2,2)),
\label{Rrr}
\\
R^{2,2}(z,q)_{\alpha,\beta}^{\gamma,\delta}
 &=  \varrho^{2,2}_{(-1)^{|\alpha|}, (-1)^{|\beta|}}(z)
\langle \chi_2 |z^{{\bf h}}
L^{\gamma_1,\delta_1}_{\alpha_1, \beta_1}
\cdots
L^{\gamma_n,\delta_n}_{\alpha_n, \beta_n}|\chi_2\rangle,
\label{R22}
\\
R^{\ast}(z,q)_{\alpha,\beta}^{\gamma,\delta}
&= R(z,q)_{\alpha,{\bf 1}-\beta}^{\gamma,{\bf 1}-\delta},
\quad
R^{\ast\ast }(z,q)_{\alpha,\beta}^{\gamma,\delta}= 
R(z,q)_{{\bf 1}-\alpha,{\bf 1}-\beta}^{{\bf 1}-\gamma,{\bf 1}-\delta}
\qquad (R= R^{\mathrm{tr}}, R^{r,r'}).
\label{r1-}
\end{align}
The normalization factors appearing here are taken as
\begin{align}
\varrho^{\mathrm {tr}}_{m,m'}(z) &=
(\I \epsilon q)^{-|m-m'|}(1-z q^{2|m-m'|}),
\\
\varrho^{r,r'}(z) &=
\frac{(z^{\max(r,r')};q^{2rr'})_\infty}
{(-z^{\max(r,r')}q^2;q^{2rr'})_\infty}\quad ((r,r') \neq (2,2)),
\\
\varrho^{2,2}_{\pm, \pm}(z)&= 
\frac{(z^2;q^8)_\infty}{(z^2q^4;q^8)_\infty},\quad 
\varrho^{2,2}_{\pm,\mp}(z)= 
\frac{\I \epsilon(z^2q^4;q^8)_\infty}{(z^2q^8;q^8)_\infty}.
\end{align}
For the definition of the symbol $|\alpha|$, see (\ref{spm}).
These choices make all the matrix elements of 
$R^{\mathrm {tr}}(z,q)$ and $R^{r,r'}(z,q)$ 
{\em rational} functions in $z$ and $q$.
For instance we have
\begin{align}
R^{\mathrm {tr}}(z,q) (
v_{{\bf e}_1+\cdots + {\bf e}_m}\otimes 
v_{{\bf e}_1+\cdots + {\bf e}_{m'}})
&= v_{{\bf e}_1+\cdots + {\bf e}_m}\otimes 
v_{{\bf e}_1+\cdots + {\bf e}_{m'}}
\quad (0 \le m, m' \le n),
\label{no1}\\
R^{r,r'}(z,q) (v_\alpha \otimes v_\alpha) 
&= v_\alpha \otimes v_\alpha
\quad (\alpha \in \mathrm{sp}),
\label{no2}\\
R^{2,2}(z,q) (v_{{\bf e}_1} \otimes v_0) 
&=v_{{\bf e}_1} \otimes v_0,
\quad
R^{2,2}(z,q)(v_0 \otimes v_{{\bf e}_1})
=v_0 \otimes v_{{\bf e}_1}.
\label{no3}
\end{align}
From the construction it is easy to see the properties:
\begin{align}
R(z,q)_{\alpha, \beta}^{\gamma, \delta} &= 0 \;\;
\text{unless}\;\; \alpha + \beta = \gamma+ \delta \quad 
(R= R^{\mathrm{tr}}, R^{r,r'}),
\label{wr}\\
R^{\mathrm{tr}}(z,q)_{\alpha, \beta}^{\gamma, \delta} &= 0 \;\;
\text{unless}\;\; |\alpha| = |\gamma|\;\;\text{and}\;\; |\beta|=|\delta|,
\\
R^{2,2}(z,q)_{\alpha, \beta}^{\gamma, \delta} &= 0 \;\;
\text{unless}\;\; |\alpha| \equiv  |\gamma|\;\;\text{and}\;\; |\beta| \equiv |\delta|
\mod 2.
\end{align}

Comparing (\ref{L}) and \cite[eq.(5)]{KP},
one can show that these matrix elements are related to 
$S(z)^{\gamma,\delta}_{\alpha,\beta}$ in 
\cite[eqs.(57), (58)]{KP} by
\begin{align}\label{epr}
R(z,q)^{\gamma,\delta}_{\alpha,\beta}
= \varphi(z) (\I \epsilon)^{|\beta-\gamma |}
S(z)^{\gamma,\delta}_{\alpha,\beta}
\qquad 
((R,S) = (R^{\mathrm{tr}}, S^{\mathrm{tr}}), (R^{r,r'}, S^{r,r'}))
\end{align}
with a scalar $\varphi(z)$ that only depends on the normalization.
Examples of $S^{\mathrm{tr}}(z), S^{r,r'}(z)$ are available in
\cite[App.C]{KP}. 
The factor $(\I \epsilon)^{|\beta-\gamma |}$ in (\ref{epr}) 
does not influence the reflection equation on account of (\ref{wr}).

\subsection{Matrix product construction from {\mathversion{bold}$G$}}
\label{ss:gk}
Define an operator $G$ by 
\begin{align}
G= 
\begin{pmatrix}
G^0_0 & G^0_1\\
G^1_0 & G^1_1
\end{pmatrix} = 
\begin{pmatrix}
\I s' \ok & \am\\
\ap & \I s \ok
\end{pmatrix},
\label{G}
\end{align}
where $s$ and $s'$ are parameters satisfying
\begin{align}\label{s12}
ss'=q.
\end{align}
The definition (\ref{G}) is a slight modification of \cite[eq.(6)]{KP}.
We introduce the linear operators
$K(z,q) = K^{\mathrm{tr}}(z,q)$ and $K^{r,r'}(z,q)$ by
\begin{align}\label{ka}
K(z,q):  V_z \rightarrow V^\ast_{z^{-1}},
\qquad
K(z,q) v_\alpha = \sum_{\beta \in \mathrm{sp}} 
K(z,q)^\beta_\alpha v^\ast_\beta.
\end{align}
Note that the dependence on $s, s'$ is suppressed in the notation.
The elements are given by the matrix product formulas:
\begin{align}
K^{\mathrm{tr}}(z,q)^\beta_\alpha 
&= \kappa^{\mathrm{tr}}_{|\alpha|}(z)
\mathrm{Tr}\bigl( z^{\h} G^{\beta_1}_{\alpha_1} 
\cdots G^{\beta_n}_{\alpha_n} \bigr),
\label{ktr}\\
K^{k,k'}(z,q)^\beta_\alpha 
&= \kappa^{k,k'}(z)
\langle \eta_k | z^{\h} G^{\beta_1}_{\alpha_1} \cdots G^{\beta_n}_{\alpha_n} 
|\eta_{k'}\rangle\quad ((k,k') \neq (2,2)),
\label{kkk}
\\
K^{2,2}(z,q)^\beta_\alpha 
&= \kappa^{2,2}_{(-1)^{|\alpha|}}(z)
\langle \eta_2 | z^{\h} G^{\beta_1}_{\alpha_1} \cdots G^{\beta_n}_{\alpha_n} 
|\eta_2\rangle.
\label{k22}
\end{align}
The operator $z^{\bf h}$ can be moved anywhere by means of 
\begin{align}\label{zG}
z^{\bf h}G^b_a = z^{b-a}\,G^b_a z^{\bf h}.
\end{align}
The normalization factors are taken as
\begin{align}
\kappa^{\mathrm{tr}}_l(z) &= (-\I)^{n}s^{-l}(s')^{l-n}(1-zq^n),
\\
\kappa^{k,k'}(z) &= (\I s')^{-n}
\frac{((q^nz)^{\max(k,k')}; q^{kk'})_\infty}
{(-q(q^nz)^{\max(k,k')}; q^{kk'})_\infty}
\quad ((k,k') \neq (2,2)),
\label{kak}\\
\kappa^{2,2}_+(z) &= \frac{s'}{s}\kappa^{2,2}_-(z) = (\I s')^{-n}
\frac{(q^{2n}z^2; q^4)_\infty}
{(q^{2n+2}z^2; q^4)_\infty}.
\end{align}
These choices make all the matrix elements of 
$K^{\mathrm {tr}}(z,q)$ and $K^{r,r'}(z,q)$ 
rational in $z$ and $q$.
For instance we have
\begin{align}
K^{\mathrm{tr}}(z,q)v_\alpha &= v^\ast_\alpha + \cdots
\quad (\alpha \in \mathrm{sp}),
\label{no4}\\
K^{r,r'}(z,q) v_0 &= v^\ast_0 + \cdots\quad ((k,k') \neq (2,2)),
\label{no5}\\
K^{2,2}(z,q)v_0 &= v^\ast_0 + \cdots, \quad 
K^{2,2}(z,q)v_{{\bf e}_1} =  v^\ast_{{\bf e}_1} + \cdots.
\label{no6}
\end{align}
By the construction they have the properties:
\begin{align}
K^{\mathrm{tr}}(z,q)_\alpha^\beta &= 0 
\;\; \text{unless}\;\;
|\alpha | = |\beta|,
\\
K^{2,2}(z,q)_\alpha^\beta &= 0 
\;\; \text{unless}\;\;
|\alpha | \equiv  |\beta| \mod 2,
\\
K^{\mathrm{tr}}(z,q)_{\alpha_1,\ldots, \alpha_n}^{\beta_1,\ldots, \beta_n}
&= K^{\mathrm{tr}}(z,q)^{\alpha_n,\ldots, \alpha_1}_{\beta_n,\ldots, \beta_1},
\\
K^{k,k'}(z,q)_{\alpha_1,\ldots, \alpha_n}^{\beta_1,\ldots, \beta_n}
&= z^{|\beta|-|\alpha|}
K^{k',k}(z,q)^{\alpha_n,\ldots, \alpha_1}_{\beta_n,\ldots, \beta_1},
\label{erk}
\end{align}
where the latter two relations can be derived from 
\cite[eq.(1)]{KP} and $\kappa^{k,k'}(z) = \kappa^{k',k}(z)$\footnote{
Eq.(\ref{erk}) corresponds to \cite[eq.(84)]{KP}.}.

Comparing (\ref{G}) and \cite[eq.(6)]{KP},
one can show that these matrix elements are related to 
\cite[eqs.(74), (83)]{KP} by
\begin{align}\label{gok}
K(z,q)^\beta_\alpha = \phi(z)
(\I q^\hf s^{-1})^{n-|\alpha|-|\beta|}
(K(z)_\alpha^{{\bf 1}-\beta} \;\text{in \cite{KP}})
\qquad (K = K^{\mathrm{tr}}, K^{k,k'})
\end{align}
with some scalar $\phi(z)$ that only depends on the normalization.
Examples of $K^{\mathrm{tr}}(z), K^{r,r'}(z)$ on the RHS are available in
\cite[App.C]{KP}. 
Setting $(s, s')=(\I t, -\I q/t)$ so as to satisfy 
(\ref{s12}), one can stay within ``real" coefficients 
$K^{\mathrm{tr}}(z,q)^\beta_\alpha, 
K^{k,k'}(z,q)^\beta_\alpha \in \Q(q,t,z)$.
This feature will also be observed explicitly in the examples
in Section \ref{ss:ex}. 

\begin{remark}\label{re:rk}
Comparing (\ref{L}) and (\ref{G}), we have
\begin{align}
G^j_i\left|_{s=s'= \epsilon q^\hf} \right. = 
L^{j,1-j}_{i,1-i}\left|_{q\rightarrow q^\hf} \right..
\end{align}
From the similar relation between the boundary vectors (\ref{ce}), 
we find a curious fact that the elements of 
the $K$ matrices (\ref{ktr})--(\ref{k22})
are equal to the special elements of the $R$ matrices (\ref{Rtr})--(\ref{r1-}) 
with $q$ replaced by $q^\hf$:
\begin{align}\label{kr}
K(z,q)_\alpha^\beta = \rho(z) 
R^\ast(z,q^\hf)^{\beta, \beta}_{\alpha, \alpha}
\qquad
((R,K) = (R^{\mathrm{tr}},K^{\mathrm{tr}}), (R^{r,r'}, K^{r,r'})),
\end{align}
where $\rho(z)$ is a scalar depending on the normalization only,
and the LHS actually 
means the case $s=s'=\epsilon q^\hf$.
\end{remark}

\subsection{Examples}\label{ss:ex}

The quantities
(\ref{ktr})--(\ref{k22}) are evaluated by the following formulas 
\cite[eq.(81)]{KP}\footnote{Quoted by taking it into account that 
 $\ok$ in \cite{KP} is equal to $q^\hf\ok$ in this paper.}:
\begin{equation}\label{lin}
\begin{split}
&\mathrm{Tr}(z^{\bf h}\ok^m) = \frac{1}{1-zq^m},
\\
&\langle \eta_k|z^{\bf h} (\apm)^j \ok^m 
w^{\bf h}|\eta_{k'}\rangle =
\langle \eta_{k'}|w^{\bf h}\ok^m  (\amp)^j 
z^{\bf h}|\eta_k\rangle\quad (k, k' \in \{1,2\}),
\\
&\langle \eta_1|z^{\bf h}(\ap)^j \ok^m 
w^{\bf h}|\eta_1\rangle
= z^j(-q;q)_j
\frac{(-q^{j+m+1}zw;q)_\infty}{(q^mzw;q)_\infty},
\\
&\langle \eta_1|z^{\bf h}(\am)^j \ok^m 
w^{\bf h}|\eta_2\rangle
= z^{-j}\sum_{i=0}^j(-1)^i q^{\frac{1}{2}i(i+1-2j)}
\binom{j}{i}_{\!\!q}
\frac{(-q^{2i+2m+1}z^2w^2;q^2)_\infty}{(q^{2i+2m}z^2w^2;q^2)_\infty},
\\
&\langle \eta_1|z^{\bf h} (\ap)^j \ok^m 
w^{\bf h}|\eta_2\rangle
= z^{j}\sum_{i=0}^j q^{\frac{1}{2}i(i+1)}
\binom{j}{i}_{\!\!q}
\frac{(-q^{2i+2m+1}z^2w^2;q^2)_\infty}{(q^{2i+2m}z^2w^2;q^2)_\infty},
\\
&\langle \eta_2|z^{\bf h} (\ap)^{j} \ok^m 
w^{\bf h}|\eta_2\rangle
= \theta(j\in 2\Z) \,z^{j}(q^2;q^4)_{\frac{j}{2}}
\frac{(q^{2j+2m+2}z^2w^2;q^4)_\infty}{(q^{2m}z^2w^2;q^4)_\infty},
\end{split}
\end{equation}
where $\theta(\text{true})=1, \theta(\text{false})=0$.

As already mentioned, 
matrix elements have been given explictly in \cite[App.C]{KP}
for $K^{\mathrm{tr}}(z,q)$ up to $n=3$
and  $K^{k,k'}(z,q)$ up to $n=2$, which are 
connected to the present paper by (\ref{gok}).
So we illustrate here the calculation of $K^{1,1}(z,q)$ with $n=3$,
which corresponds to $U_p(\mathfrak{g}^{1,1}) = U_p(D^{(2)}_4)$.

From (\ref{erk}) and (\ref{zG}) with $z^{\bf h}$ replaced by 
the special case $\ok=q^{\bf h}$,
the elements\footnote{For simplicity, 
$K^{1,1}(z,q)^\beta_\alpha$ with 
$\alpha=(\alpha_1, \alpha_2, \alpha_3),
\beta=(\beta_1, \beta_2, \beta_3)$ is denoted by 
$K^{1,1}(z,q)^{\beta_1\beta_2\beta_3}_{\alpha_1\alpha_2\alpha_3}$
rather than 
$K^{1,1}(z,q)^{(\beta_1,\beta_2,\beta_3)}_{(\alpha_1,\alpha_2,\alpha_3)}$.}
are reduced to
$K^{1,1}(z,q)^{000}_{000}=1$ in (\ref{no5}) 
and the following up to powers of $q, s$ and $z$:
\begin{alignat*}{2}
K^{1,1}(z,q)_{100}^{001} 
&= -\frac{s^2(1+q)(1-q+q^2z+q^3z)}{q(qz;q)_2},
&\quad
K^{1,1}(z,q)_{010}^{100}
&=-\frac{s^2(1+q)z(1+q-q^2z+q^3z)}{q(qz;q)_2},
\\
K^{1,1}(z,q)_{001}^{110} &= \frac{\I s^3 (1+q)^2(1+q^2)z^2(1-qz+q^2z)}{q^3(z;q)_3},
&\quad
K^{1,1}(z,q)_{000}^{011} &= -\frac{s^2(1+q)(1+q^2)z^2}{(1-qz)(1-q^2z)},
\\
K^{1,1}(z,q)_{010}^{101} 
&= \frac{\I s^3 (1+q)^2z(1-q+2q^2z-q^3z^2+q^4z^2)}{q^3(z;q)_3},
&\quad
K^{1,1}(z,q)_{000}^{001} &= -\frac{\I s q(1+q)z}{1-q^2z},
\\
K^{1,1}(z,q)_{000}^{111} &= \frac{\I s^3 (-q;q)_3z^3}{q^3(z;q)_3},
&\quad
K^{1,1}(z,q)_{100}^{011} &= \frac{\I s^3 (1+q)^2(1+q^2)z(1-q+q^2z)}{q^3(z;q)_3}.
\end{alignat*}
Let us derive $K^{1,1}(z,q)_{100}^{011}$ in the last line.
From (\ref{kkk}) and $s'=q/s$ it is calculated as
\begin{align*}
&\kappa^{1,1}(z)\langle \eta_1| z^{\bf h}G^0_1G^1_0G^1_0 |\eta_1\rangle
= \kappa^{1,1}(z)\langle \eta_1| z^{\bf h}\am \ap \ap |\eta_1\rangle
= \kappa^{1,1}(z)\langle \eta_1| z^{\bf h}(1-q^2\ok^2) \ap |\eta_1\rangle
\\
&=\kappa^{1,1}(z)\langle \eta_1| z^{\bf h}(\ap-q^4 \ap\ok^2)  |\eta_1\rangle
=\frac{\I s^3 (q^3z;q)_\infty}{q^3(-q^4z;q)_\infty}
\left(\frac{z(1+q)(-q^2z;q)_\infty}{(z;q)_\infty}
-\frac{q^4z(1+q)(-q^4z;q)_\infty}{(q^2z;q)_\infty}\right)
\\
&=\frac{\I s^3z(1+q)}{q^3}
\left(\frac{(1+q^2z)(1+q^3z)}{(z;q)_3}-\frac{q^4}{1-q^2z}\right),
\end{align*}
which leads to the sought result.

\section{Main result}\label{sec:mr}

In the rest of the paper we will always assume that 
$\epsilon, q, p, s, s'$ are related by
\begin{align}\label{qp}
\epsilon= \I qp \in \{1,-1\},
\qquad ss' = q.
\end{align}
This unifies (\ref{s12}) and (\ref{ss12}).

\subsection{Summary of known results}
Let us first recall the known result on the $R$ matrices.
\begin{theorem}\label{th:r}
The $R$ matrices possess the matrix product formulas as follows:
\begin{align}
U_p(\mathfrak{g}^{\mathrm{tr}}):\;
\mathscr{R}(z,p)
& = R^{\mathrm{tr}}(z,q),
\label{bs}\\
U_p(\mathfrak{g}^{r,r'}):\;
\mathscr{R}(z,p)& = R^{r,r'}(z,q).
\label{ks}
\end{align}
\end{theorem}
The results (\ref{bs}) and (\ref{ks}) are essentially due to
\cite{BS} and \cite{KS}, respectively.

From (\ref{rmm}) and (\ref{r1-}), Theorem \ref{th:r} implies
\begin{align} 
U_p(\mathfrak{g}^{\mathrm{tr}}):\;
\mathscr{R}^\ast(z,p)
&= R^{\mathrm{tr}\ast}(z,q),
\quad\;\;
\mathscr{R}^{\ast\ast}(z,p)
= R^{\mathrm{tr}\ast\ast}(z,q),
\\
U_p(\mathfrak{g}^{r,r'}):\;
\mathscr{R}^\ast(z,p)&= R^{r,r'\ast}(z,q),
\quad
\mathscr{R}^{\ast\ast}(z,p)
= R^{r,r'\ast\ast}(z,q).
\end{align}

Next we summarize the results on the $K$ matrices in \cite{KP}.
For $U_p(\mathfrak{g}^{\mathrm{tr}})$,
the reflection equation involving 
$R^{\mathrm{tr}}(z,q), K^{\mathrm{tr}}(z,q)$ was proved in \cite[eq.(73)]{KP}.
For $U_p(\mathfrak{g}^{r,r'})$,
the reflection equation involving 
$R^{r,r'}(z,q)$ and $K^{k,k'}(z,q)$
with $r\le k, r'\le k'$ was shown to follow from 
the property \cite[eq.(78)]{KP} of the boundary vectors. 
Proving it is an interesting open problem.
In this paper we will achieve a proof of 
the reflection equation for $U_p(\mathfrak{g}^{r,r'})$ 
via a different route of resorting to the 
representation theory of the coideal subalgebras of $U_p$.

\subsection{{\mathversion{bold}$K$} matrices}
Now we state the main results of the paper.
They characterize the $K$ matrices \cite{KP} by the matrix product construction 
(recalled in Section \ref{ss:gk} with a slight parametric generalization) 
as the intertwiner of the coideal subalgebras
described in Section \ref{ss:coi}--\ref{ss:cok}.

\begin{theorem}\label{th:A}
The $U_p(\mathfrak{g}^{\mathrm{tr}})$ modules $V^m_x$ and
$V^m_x \otimes V^{m'}_y$ with generic $x,y$ are 
irreducible as $\mathcal{B}^{\mathrm{tr}}$ modules for any $0\le m, m' \le n$.
There is a $K$ matrix having the decomposition
\begin{align}\label{cdk}
\mathscr{K}(z,p) = \mathscr{K}^0(z,p) \oplus \cdots \oplus \mathscr{K}^n(z,p),
\quad
\mathscr{K}^m(z,p): V^m_z \rightarrow V^{m \ast}_{z^{-1}}
\end{align}
such that each component is characterized up to normalization 
by a finer version of the 
intertwining relation (\ref{ir1}): 
\begin{align}
\mathscr{K}^m(z,p) \pi^m_z(b) &= \pi^{m\ast}_{z^{-1}}(b) \mathscr{K}^m(z,p)
\quad (b \in \mathcal{B}^{\mathrm{tr}}).
\end{align}
It is given by the matrix product formula
\begin{align}\label{ttk}
\mathscr{K}(z,p) = K^{\mathrm{tr}}(z,q).
\end{align}
\end{theorem}

\begin{theorem}\label{th:BD2}
For $(r,r')\neq (2,2)$, 
the $U_p(\mathfrak{g}^{r,r'})$ modules $V_x$ and
$V_x \otimes V_y$ with generic $x,y$ are 
irreducible as $\mathcal{B}^{r,r'}_{k,k'}$ modules 
if $r\le k$ and $r' \le k'$.
The $K$ matrix characterized up to normalization by  
the intertwining relation (\ref{ir1}) 
is given by the matrix product formula:
\begin{align}
\mathscr{K}(z,p) = K^{k,k'}(z,q).
\end{align}
\end{theorem}

\begin{theorem}\label{th:D}
The $U_p(\mathfrak{g}^{2,2})$ modules $V^\pm_x$, 
$V^\pm_x \otimes V^\pm_y$ and $V^\pm_x \otimes V^\mp_y$ with generic $x,y$ are 
irreducible as $\mathcal{B}^{2,2}_{2,2}$ modules.
There is a $K$ matrix having the decomposition
\begin{align}
\mathscr{K}(z,p) = \mathscr{K}^+(z,p) \oplus \mathscr{K}^-(z,p),
\quad
\mathscr{K}^\pm(z,p): V^\pm_z \rightarrow V^{\pm \ast}_{z^{-1}}
\end{align}
such that each component is characterized  up to normalization 
by a finer version of the intertwining relation (\ref{ir1}):
\begin{align}
\mathscr{K}^\pm(z,p) \pi^\pm_z(b) &= \pi^{\pm\ast}_{z^{-1}}(b) \mathscr{K}^\pm(z,p)
\quad (b \in \mathcal{B}^{2,2}_{2,2}).
\end{align}
It is given by the matrix product formula
\begin{align}
\mathscr{K}(z,p) = K^{2,2}(z,q).
\end{align}
\end{theorem}

\subsection{Reflection equation}
Write ${R}(z,q), {R}^\ast(z,q), 
{R}^{\ast\ast}(z,q), {K}(z,q)$ 
simply as ${R}(z), {R}^\ast(z), 
{R}^{\ast\ast}(z), {K}(z)$,
respectively.
From the argument that led to (\ref{re1}) 
and Theorem \ref{th:A}--\ref{th:D}, we have
\begin{corollary}
The following reflection equation holds:
\begin{align}
K^{\mathrm{tr}}_2(y)R^{\mathrm{tr}\ast}_{21}(xy)
K^{\mathrm{tr}}_1(x)R^{\mathrm{tr}}_{12}(x/y)
&=
R^{\mathrm{tr}\ast\ast }_{21}(x/y)
K^{\mathrm{tr}}_1(x)R^{\mathrm{tr}\ast }_{12}(xy)K^{\mathrm{tr}}_2(y),
\label{re2}\\
K^{k,k'}_2(y)R^{r,r'\ast}_{21}(xy)K^{k,k'}_1(x)R^{r,r'}_{12}(x/y)
&= 
R^{r,r'\ast\ast }_{21}(x/y)K^{k,k'}_1(x)R^{r,r'\ast }_{12}(xy)K^{k,k'}_2(y),
\label{re3}
\end{align}
where $r\le k$ and $r'\le k'$ in the latter. 
The equality (\ref{re2}) splits into the one in 
$\mathrm{Hom}(V^m_x\otimes V^{m'}_y, V^{m\ast}_{x^{-1}}
\otimes V^{m'\ast}_{y^{-1}})$
for each pair $(m,m') \in [0,n]^2$.
The equality (\ref{re3}) with $(r,r')=(2,2)$ splits into the one in 
$\mathrm{Hom}(V^\sigma_x\otimes V^{\sigma'}_y, 
V^{\sigma\ast}_{x^{-1}}\otimes V^{\sigma'\ast}_{y^{-1}})$
for each pair $(\sigma, \sigma') \in \{+,-\}^2$.
\end{corollary}

The result (\ref{re2}) reconfirms \cite[eq.(73)]{KP}.
The result (\ref{re3}) yields the first proof of \cite[eq.(82)]{KP}.
As for the special components $\mathscr{K}^1(z,q)$ and 
$\mathscr{K}^{n-1}(z,q)$ in (\ref{cdk}),
an explicit list of the matrix elements was given in \cite{G}.  

\section{Proof}\label{sec:p}

\subsection{Existence}
Let us prove the existence of the intertwiner of the coideal 
$\mathscr{K}(z,p)$ in (\ref{ir1}).
Our strategy is to show that 
the matrix product constructed $K$ matrices in Section \ref{ss:gk}
indeed fulfill the intertwining relation.
We assume the relation (\ref{qp}) among the parameters throughout.

\begin{proposition}\label{pr:kpai}
The $K$ matrices given 
by the matrix product construction (\ref{ktr})--(\ref{k22})
satisfy the intertwining relation 
\begin{align}
K(z,q) \pi_z(b) &= \pi^\ast_{z^{-1}}(b) K(z,q)
\qquad (b \in \mathcal{B}),
\label{kpi}
\end{align}
where 
$(K,\mathcal{B}) = (K^{\mathrm{tr}}, \mathcal{B}^{\mathrm{tr}})$
for $U_p(\mathfrak{g}^{\mathrm{tr}})$ and 
$(K,\mathcal{B}) = (K^{k,k'}, \mathcal{B}^{r,r'}_{k,k'})$
for $U_p(\mathfrak{g}^{r,r'})$ with $r\le k$ and $r' \le k'$.
\end{proposition}
\begin{proof}
It suffices to show (\ref{kpi}) for $b=b_0,\ldots, b_{n'}$ 
given by (\ref{bai})--(\ref{drk}).
There are seven cases to consider;
(i) $b=b_j$  (\ref{bai}) for $j \in \Z_n$,
(ii) $b=b_0$ (\ref{b0}) for $(r,k)=(1,2)$,
(iii) $b=b_0$ (\ref{b0}) for $(r,k)=(1,1)$,
(iv) $b=b_0$ (\ref{b0}) for $(r,k)=(2,2)$,
(v) $b=b_n$ (\ref{bn}) for $(r,k)=(1,2)$,
(vi) $b=b_n$ (\ref{bn}) for $(r,k)=(1,1)$,
(vii) $b=b_n$ (\ref{bn}) for $(r,k)=(2,2)$.
It is easy to show that 
(v), (vi), (vii) are attributed respectively to (ii), (iii), (iv) with $z=1$ 
by the exchange $s \leftrightarrow s'$, $G^j_i \leftrightarrow G^i_j$
and $\langle m | \leftrightarrow |m\rangle$.
See also \cite[eq.(1)]{KP} for a hint.
Thus we shall only treat the cases (i)--(iv) in the sequel.

(i) $b=b_j$  (\ref{bai}) for $j \in \Z_n$.
This covers $U_p(\mathfrak{g}^{\mathrm{tr}})$ 
whose sub-case $0<j<n$ also does the other $U_p(\mathfrak{g}^{r,r'})$.
In terms of the matrix element for the transition $v_\alpha \mapsto v^\ast_\beta$,
the equation (\ref{kpi}) is stated as
\begin{equation}
\begin{split}
&z^{-\delta_{j0}}K(z,q)_{\alpha+{\bf e}_j - {\bf e}_{j+1}}^\beta
+ z^{\delta_{j0}}p^{2(\alpha_j - \alpha_{j+1}-1)}
K(z,q)_{\alpha-{\bf e}_j + {\bf e}_{j+1}}^\beta
+ \frac{\I \epsilon p}{1-p^2}p^{2(\alpha_j-\alpha_{j+1})}K(z,q)_\alpha^\beta
\\
&=z^{\delta_{j0}}K(z,q)_\alpha^{\beta+{\bf e}_j - {\bf e}_{j+1}}
+z^{-\delta_{j0}}
p^{2(\beta_{j+1}-\beta_j+1)}K(z,q)_\alpha^{\beta-{\bf e}_j + {\bf e}_{j+1}}
+ \frac{\I \epsilon p}{1-p^2}p^{2(\beta_{j+1}-\beta_{j})}K(z,q)_\alpha^\beta,
\end{split}
\end{equation}
where $K=K^{\mathrm{tr}}\,(j \in \Z_n)$ or 
$K = K^{k,k'}\,(0<j<n)$.
In view of the matrix product formulas (\ref{ktr})--(\ref{k22}),
this follows from the quadratic relation 
(we set $(\alpha_j,\alpha_{j+1},\beta_j,\beta_{j+1}) = (a,a',b,b')$)
\begin{equation}\label{gg}
\begin{split}
&G^b_{a+1}G^{b'}_{a'-1}+ p^{2(a-a'-1)}G^b_{a-1}G^{b'}_{a'+1}
+ \frac{\I \epsilon p^{2(a-a')+1}}{1-p^2}G^b_{a}G^{b'}_{a'}
\\
&= G^{b+1}_aG^{b'-1}_{a'} + p^{2(b'-b+1)}G^{b-1}_aG^{b'+1}_{a'}
+ \frac{\I \epsilon p^{2(b'-b)+1}}{1-p^2}G^b_{a}G^{b'}_{a'},
\end{split}
\end{equation}
where we regard $G^b_a=0$ unless $a,b \in \{0,1\}$.
The factor $z^{\pm \delta_{j0}}$ has disappeared 
owing to (\ref{zG}).
It is elementary to check (\ref{gg}) by using (\ref{qb}).

(ii) $b=b_0$ (\ref{b0}) for $(r,k)=(1,2)$.
We are to check 
\begin{align}
z^{-1}K^{2,k'}(z,q)_{\alpha-{\bf e}_1}^\beta +
\frac{s'p^{-2\alpha_1}}{s}zK^{2,k'}(z,q)_{\alpha+{\bf e}_1}^\beta
= z K^{2,k'}(z,q)_\alpha^{\beta-{\bf e}_1}
+ \frac{s'p^{2\beta_1}}{s}z^{-1}K^{2,k'}(z,q)_\alpha^{\beta+{\bf e}_1}.
\end{align}
Due to (\ref{kkk}) and (\ref{k22}) this follows from
\begin{align}
\langle \eta_2|z^{\bf h}\bigl(
z^{-1}G^b_{a-1}+\frac{s'}{s}zG^b_{a+1}\bigr)
= \langle \eta_2|z^{\bf h}\bigl(
zG^{b-1}_a+\frac{s'}{s}z^{-1}G^{b+1}_a\bigr).
\end{align}
Thanks to (\ref{zG}) this further reduces to $z=1$,
which can be checked easily by (\ref{et2}).

(iii) $b=b_0$ (\ref{b0}) for $(r,k)=(1,1)$.
We are to check
\begin{equation}
\begin{split}
&z^{-1}K^{1,k'}(z,q)_{\alpha-{\bf e}_1}^\beta +
\frac{s'p^{-2\alpha_1}}{s}zK^{1,k'}(z,q)_{\alpha+{\bf e}_1}^\beta
+ d_{1,1}(s) p^{1-2\alpha_1}K^{1,k'}(z,q)_\alpha^\beta
\\
&= z K^{1,k'}(z,q)_\alpha^{\beta-{\bf e}_1}
+ \frac{s'p^{2\beta_1}}{s}z^{-1}K^{1,k'}(z,q)_\alpha^{\beta+{\bf e}_1}
+ d_{1,1}(s) p^{2\beta_1-1}K^{1,k'}(z,q)_\alpha^\beta.
\end{split}
\end{equation}
Due to (\ref{kkk}) this follows from
\begin{align}
\langle \eta_1|z^{\bf h}\bigl(
z^{-1}G^b_{a-1}+\frac{s'}{s}zG^b_{a+1}
+ d_{1,1}(s)p^{1-2a}G^b_a \bigr)
= \langle \eta_1|z^{\bf h}\bigl(
zG^{b-1}_a+\frac{s'}{s}z^{-1}G^{b+1}_a
+ d_{1,1}(s)p^{2b-1}G^b_a \bigr).
\end{align}
Due to (\ref{zG}) this reduces to $z=1$, 
which can be checked by (\ref{et1}).

(iv) $b=b_0$ (\ref{b0}) for $(r,k)=(2,2)$.
We are to check
\begin{equation}
\begin{split}
&z^{-2}K^{2,k'}(z,q)_{\alpha-{\bf e}_1-{\bf e}_2}^\beta +
\Bigl(\frac{s'z}{s}\Bigr)^2K^{2,k'}(z,q)_{\alpha+{\bf e}_1+{\bf e}_2}^\beta
+ d_{2,2}(s) p^{2(1-\alpha_1-\alpha_2)}K^{2,k'}(z,q)_\alpha^\beta
\\
&= z^2 K^{2,k'}(z,q)_\alpha^{\beta-{\bf e}_1-{\bf e}_2}
+ \Bigl(\frac{s'}{sz}\Bigr)^2 K^{2,k'}(z,q)_\alpha^{\beta+{\bf e}_1+{\bf e}_2}
+ d_{2,2}(s) p^{2(\beta_1+\beta_2-1)}K^{2,k'}(z,q)_\alpha^\beta.
\end{split}
\end{equation}
Due to (\ref{kkk}) and (\ref{k22}) this follows from
\begin{align}
&\langle \eta_2|z^{\bf h}\Bigl(
z^{-2}G^b_{a-1}G^{b'}_{a'-1}
+\Bigl(\frac{s'z}{s}\Bigr)^2G^b_{a+1}G^{b'}_{a'+1}
+ d_{2,2}(s)p^{2(1-a-a')}G^b_aG^{b'}_{a'} \Bigr)
\\
&= \langle \eta_2|z^{\bf h}\Bigl(
z^2G^{b-1}_aG^{b'-1}_{a'}
+\Bigl(\frac{s'}{sz}\Bigr)^2G^{b+1}_aG^{b'+1}_a
+ d_{2,2}(s)p^{2(b+b'-1)}G^b_aG^{b'}_{a'} \Bigr).
\end{align}
Again this reduces to $z=1$ case, which can be verified by (\ref{et2})
and (\ref{qb}).
\end{proof}

\subsection{Uniqueness and the proof of the reflection equation}

In this subsection we give a proof of the irreducibility of $V_{x}\ot V'_{y}$ as a 
module over $\mathcal{B}$ for various $\mathcal{B}$ in Section \ref{ss:coi} and $V,V'$ 
in Section \ref{ss:rep}. Once this irreducibility is proven, the proof of the reflection equation
\eqref{re1} is complete. If we consider $V'_y$ to be a trivial module, the uniqueness 
of the $K$ matrix in Section \ref{ss:cok} is also confirmed.

To show $V_x\ot V'_y$ is irreducible for generic values of $x,y$, it suffices to show it 
is so at a special value $y=x$. 
We consider $U_p$ or $\mathcal{B}$ as algebras over $\C(p)$
and $V_x,V'_x$ as modules over $\C(p,x)$.
Define $A=\{f(p,x)\in\C(p,x)\mid f(p,x)\text{ is regular at }x=0\}$. $A$ is a local ring
with the maximal ideal $xA$ and $A/xA\simeq\C(p)$.

Let $W$ be a nonzero subspace of $V_x\ot V'_x$ invariant under $\mathcal{B}$ and set
\begin{equation}
L=\bigoplus_{\alpha,\beta}Av_{\alpha}\ot v_{\beta},\quad L'=W\cap L.
\end{equation}
Apparently, $L$ is a finitely generated module over $A$ and $L'$ is a sub $A$-module
of $L$. 
Set $\tilde{b}_i=x_0^{\delta_{i0}}b_i$,
where $x_0=x^2$ 
for $U_p(\mathfrak{g}^{2,1})$, $U_p(\mathfrak{g}^{2,2})$ and 
$x_0=x$ for the other $U_p$. 
Note that the action of $\Delta(\tilde{b}_i)$ preserves $L$. 

Let $U_{p,0}$ be the subalgebra of $U_p$ generated by $e_i,f_i,k_i$ with $i\ne0$. 
Suppose that as a $U_{p,0}$-module we have the following decomposition.
\begin{equation} \label{decomp}
V_{x}\ot V'_{x}\simeq \bigoplus_{j=0}^lU_j
\end{equation}
Here we number $\{U_j\}_{0\le j\le l}$ in such a way that $j<k$ holds 
whenever $\lambda_j>\lambda_k$ occurs\footnote{
$\lambda > \mu$ is defined by $\lambda \neq \mu$ and 
$\lambda - \mu$ is a linear combination of simple roots with 
nonnegative coefficients.} 
where $\lambda_j$ is the highest weight of $U_j$.
Take any nonzero vector $u$ from $L'$. 
According to the decomposition \eqref{decomp},
$u$ can be expressed as
\begin{equation} \label{imozuru}
u=\sum_{j,\gamma}u_{j,\gamma}
\end{equation}
where $u_{j,\gamma}$ is a weight vector of weight $\gamma$ which belongs to
$U_j$. One can assume there exists $(j,\gamma)$ such that 
$u_{j,\gamma}\not\equiv0$ mod $xL$. 
Here and in what follows, $\equiv$ means an equality mod $xA$.
Take the minimal $j_0$ such that $u_{j_0,\gamma}\not\equiv0$ for some $\gamma$.
Applying $\Delta(\tilde{b}_i)$ ($i\ne0$) if necessary, one can assume 
$u_{j_0,\gamma_0}\not\equiv0$ where $\gamma_0$ is the lowest weight of 
$U_{j_0}$.

\begin{lemma}\label{lem:1}
Suppose $j>0$. Then there exists a sequence $i_1,\ldots,i_m\in\{0,1,\ldots,n'\}$ such
that $\Delta(\tilde{b}_{i_1}\cdots\tilde{b}_{i_m})u$ 
has a nonzero component in $U_{j'}$
($j'<j$).
\end{lemma}

\begin{proof}
Since the proof depends on the cases of $U_p$ and $\mathcal{B}$, we illustrate it for 
(i) $U_p(\mathfrak{g}^{\mathrm{tr}})$ 
with $n=4$ and $V^m_x\ot V^{m'}_x$ with $m=m'=2$,
and (ii) $U_p(\mathfrak{g}^{1,1})$ with $n=2$ and $V_x\ot V_x$ (spin representation).

(i) In this case the decomposition \eqref{decomp} reads as
\begin{equation}
V^2_x\ot V^2_x\simeq U_0\oplus U_1\oplus U_2.
\end{equation}
Suppose $j=1$. One can assume \eqref{imozuru} is of the following form:
\begin{equation}
u\equiv(v_{(1100)}\ot v_{(1010)}-p^2v_{(1010)}\ot v_{(1100)})+\cdots.
\end{equation}
Here $\cdots$ contains a linear combination of weight vectors of weight strictly higher than
that of $v_{(1100)}\ot v_{(1010)}$. Noting that 
$\tilde{b}_0=x_0f_0$ mod
$xL'$. We have
\begin{equation}
\Delta(\tilde{b}^2_0)u\equiv [2]_{p^2}(v_{(0101)}\ot v_{(0011)}-p^2v_{(0011)}\ot v_{(0101)})+
\cdots
\end{equation}
It may happen that $\cdots$, terms with higher weights, contain a nonzero component 
in $U_0$. In this case the proof is over. Suppose they do not. Applying $\tilde{b}_3
\tilde{b}_1$ further, we get
\begin{equation}
\Delta(\tilde{b}_3\tilde{b}_1\tilde{b}^2_0)u\equiv 
[2]_{p^2}(v_{(1010)}\ot v_{(0011)}-p^2v_{(0011)}\ot v_{(1010)})+\cdots.
\end{equation}
Note that $\cdots$ contain terms produced by applying $k_i^{-1}e_i$ or $k_i^{-1}$
in $\tilde{b}_i$ for $i=1,3$. However the weights remain higher than that of 
$v_{(1010)}\ot v_{(0011)}$. Finally, applying $\tilde{b}_0$, we obtain 
\begin{equation}
\Delta(\tilde{b}_0\tilde{b}_3\tilde{b}_1\tilde{b}^2_0)u\equiv 
[2]_{p^2}(1-p^4)v_{(0011)}\ot v_{(0011)}+\cdots.
\end{equation}
The first term apparently belongs to $U_0$.

Suppose $j=2$. One can assume \eqref{imozuru} is of the following form:
\begin{align} \label{lwv1}
u\equiv v_{(1100)}\ot v_{(0011)}&-p^2v_{(1010)}\ot v_{(0101)}
+p^4(v_{(1001)}\ot v_{(0110)}+v_{(0110)}\ot v_{(1001)}) \nonumber \\
&-p^6v_{(0101)}\ot v_{(1010)}+p^8v_{(0011)}\ot v_{(1100)}.
\end{align}
Since $\dim U_0=1$, There are no other terms. Applying $\tilde{b}_0$ we have
\begin{equation}
\Delta(\tilde{b}_0)u\equiv (1-p^8)
(v_{(0101)}\ot v_{(0011)}-p^2v_{(0011)}\ot v_{(0101)}).
\end{equation}
Since it is not proportional to \eqref{lwv1}, it should contain a nonzero component in
$U_0$ or $U_1$.

(ii) Since the situation is similar, we only list formulas. For $j=1$ \eqref{imozuru} is 
of the form:
\begin{equation}
u\equiv v_{(11)}\ot v_{(10)}-pv_{(10)}\ot v_{(11)}+\cdots
\end{equation}
and we have
\begin{equation}
\Delta(\tilde{b}_0\tilde{b}_1\tilde{b}_0^2)u\equiv
[2]_p(1-p^2)v_{(00)}\ot v_{(00)}+\cdots.
\end{equation}
For $j=2$ \eqref{imozuru} is of the form:
\begin{equation}
u\equiv v_{(11)}\ot v_{(00)}-pv_{(10)}\ot v_{(01)}+p^3v_{(01)}\ot v_{(10)}-p^4v_{(00)}\ot v_{(11)}
\end{equation}
and we have
\begin{equation}
\Delta(\tilde{b}_0)u\equiv
(1+p^4)(v_{(01)}\ot v_{(00)}-pv_{(00)}\ot v_{(01)}).
\end{equation}
\end{proof}

From this lemma, one can assume $j_0=0$ and $\gamma_0$ is the lowest weight of
$U_0$. We can also show the following lemma.

\begin{lemma}
For a vector $u=\sum_{j,\gamma}u_{j,\gamma}$ 
such that $u_{0,\gamma_0}\not\equiv0$
where $\gamma_0$ is the lowest weight of $U_0$, there exists a sequence 
$i_1,\ldots,i_m\in\{0,1,\ldots,n'\}$ such that 
\[
\Delta(\tilde{b}_{i_1}\cdots\tilde{b}_{i_m})u\equiv u_0
\]
where $u_0$ is a nonzero highest weight vector of $U_0$.
\end{lemma}

\begin{proof}
As in the proof of Lemma \ref{lem:1}, we illustrate the proof for 
(i) $U_p(\mathfrak{g}^{\mathrm{tr}})$ with $n=4$ and $V^m_x\ot V^{m'}_x$ with $m=m'=2$,
and (ii) $U_p(\mathfrak{g}^{1,1})$ with $n=2$ and $V_x\ot V_x$ (spin representation).

(i) The vector $u$ is of the form:
\begin{equation}
u\equiv v_{(1100)}\ot v_{(1100)}+\cdots,
\end{equation}
where $\cdots$ contains a linear combination of weight vectors of weight strictly higher 
than that of $v_{(1100)}\ot v_{(1100)}$. Applying
$\tilde{b}_0^2\tilde{b}_3^2\tilde{b}_1^2\tilde{b}_0^2$, we have
\begin{equation}
\Delta(\tilde{b}_0^2\tilde{b}_3^2\tilde{b}_1^2\tilde{b}_0^2)u\equiv
[2]_{p^2}^4v_{(0011)}\ot v_{(0011)}+\cdots.
\end{equation}
Since $v_{(0011)}\ot v_{(0011)}$ is a highest weight vector, the part $\cdots$ should 
vanish.

(ii) The vector $u$ is of the form:
\begin{equation}
u\equiv v_{(11)}\ot v_{(11)}+\cdots,
\end{equation}
where $\cdots$ contains a linear combination of weight vectors of weight strictly higher 
than that of $v_{(11)}\ot v_{(11)}$. Applying
$\tilde{b}_0^2\tilde{b}_1^2\tilde{b}_0^2$, we have
\begin{equation}
\Delta(\tilde{b}_0^2\tilde{b}_1^2\tilde{b}_0^2)u\equiv
[2]_{p}^2[2]_{p^2}v_{(00)}\ot v_{(00)}+\cdots.
\end{equation}
Since $v_{(00)}\ot v_{(00)}$ is a highest weight vector, the part $\cdots$ should 
vanish.
\end{proof}

Hence, one finds that $u_0$
is contained in $L'$ mod $xL'$. Now apply $\Delta(\tilde{b}_i)$ successively
to obtain new vectors in $L'$. Since new vectors obtained are those by applying
$\Delta(x_0^{\delta_{i0}}f_i)$, $u_0$ generates all vectors in $L$ mod
$xL$, so we have $L=L'+xL$. By Nakayama's lemma (cf. \cite{AM}):
\begin{lemma}
Suppose that $R$ is a ring and $M$ is an $R$-module. Let $N$ be a submodule of $M$. 
If $M/N$ is finitely generated over $R$ and $M=N+J(R)M$, then $N=M$.
Here $J(R)$ is the Jacobson radical of $R$.
\end{lemma}

\noindent
with $R=A$ ($J(R)=xA$) and $M=L,N=L'$, 
one obtains $L'=L$, and hence $W=V_{x}\ot V'_{x}$. 

The proof of the irreducibility of $V_x\ot V'_y$ is complete.

\section{Generalization}\label{s:gen}

\subsection{Inhomogeneous matrix product}
The parameters $s,s'$ entering (\ref{G}) may be 
taken independently in each $G$ in the matrix products
(\ref{ktr})--(\ref{k22})
as long as the relation (\ref{s12}) is kept.
To be concrete, we write (\ref{G}) obeying the constraint (\ref{s12}) as  
\begin{align}
\begin{pmatrix}
G(s)^0_0 & G(s)^0_1\\
G(s)^1_0 & G(s)^1_1
\end{pmatrix} = 
\begin{pmatrix}
\I qs^{-1} \ok & \am\\
\ap & \I s \ok
\end{pmatrix}.
\end{align}
Introduce the parameters 
$s_1, \ldots, s_n$ 
and consider the inhomogeneous generalization of
(\ref{ktr})--(\ref{k22}):
\begin{align}
K^{\mathrm{tr}}(z,q; s_1,\ldots, s_n)^\beta_\alpha 
&= \mathrm{Tr}\bigl( z^{\h} G(s_1)^{\beta_1}_{\alpha_1} 
\cdots G(s_n)^{\beta_n}_{\alpha_n} \bigr),
\label{iktr}\\
K^{k,k'}(z,q; s_1,\ldots, s_n)^\beta_\alpha 
&= \langle \eta_k | z^{\h} G(s_1)^{\beta_1}_{\alpha_1} 
\cdots G(s_n)^{\beta_n}_{\alpha_n} 
|\eta_{k'}\rangle
\label{ikkk},
\end{align}
where the normalization factors have been put aside for simplicity.
Using them as the matrix elements 
define $K^{\mathrm{tr}}(z,q; s_1,\ldots, s_n),
K^{k,k'}(z,q; s_1,\ldots, s_n) \in \mathrm{Hom}(V_z, V^\ast_{z^{-1}})$
similarly to (\ref{ka}).
The homogeneous case (\ref{ka})--(\ref{k22})
corresponds to $K(z,q,s,\ldots, s)$. 
The general case is reduced to it as
\begin{align}\label{KS}
S^{\ast }K(z,q; s_1,\ldots, s_n) S^{-1}= K(z,q; s,\ldots, s)
\qquad (K = K^{\mathrm{tr}}, K^{k,k'}),
\end{align}
where $S$ and $S^\ast$ are the diagonal matrices defined by 
\begin{align}\label{S}
S=\mathrm{diag}(S_\alpha)_{\alpha \in \mathrm{sp}},
\quad
S^\ast=\mathrm{diag}(S_{{\bf 1}-\alpha})_{\alpha \in \mathrm{sp}},
\quad S_\alpha = \prod_{j=1}^n(s_j/s)^{\alpha_j}.
\end{align}
As for the $R$ matrices, 
the weight conservation (\ref{wr}) and the relation (\ref{r1-}) imply 
\begin{align}\label{rd}
[R(z,q), S \otimes S] = 0,
\quad
[R^\ast(z,q), S \otimes S^\ast]=0,
\quad
[R^{\ast\ast}(z,q), S^\ast \otimes S^\ast]=0. 
\end{align}
From (\ref{KS}) and (\ref{rd}) 
it follows that $K(z,q; s_1,\ldots, s_n)$ also satisfies the 
same reflection equations as (\ref{re2}) and (\ref{re3}) 
without affecting the companion $R$ matrices.
In this sense $s_1, \ldots, s_n$ in  (\ref{iktr}) and 
(\ref{ikkk}) are just a sort of gauge parameters
as far as the reflection equation is concerned.

On the other hand their influence on the intertwining relation (\ref{kpi}) 
is less trivial. In fact it is transformed to
\begin{align}
K(z,q; s_1,\ldots, s_n)S^{-1}\pi_z(b)S = 
S^{\ast -1}\pi^\ast_{z^{-1}}(b) S^\ast
K(z,q; s_1,\ldots, s_n)\qquad (b \in \mathcal{B}).
\end{align}
One may interpret this as the 
intertwining relation for the representation 
connected to $\pi_z$ 
by the conjugation by $S$.
However there is another option 
to stay within the same representation 
saying that it is the coideal $\mathcal{B}=\langle b_0, \ldots, b_{n'}\rangle$
that has been deformed by the parameters $s_1, \ldots, s_n$. 
The latter is presented as 
$\mathcal{B}(s_1,\ldots, s_n)
= \langle \hat{b}_0, \ldots, \hat{b}_{n'}\rangle$
in terms of the new generators $\hat{b}_i$ satisfying 
\begin{align}
S^{-1}\pi_z(b_i)S = \pi_z(\hat{b}_i),
\qquad
S^{\ast -1}\pi^\ast_{z^{-1}}(b_i)S^\ast  = \pi^\ast_{z^{-1}}(\hat{b}_i).
\end{align}
Below we describe such $\hat{b}_i$'s generating 
$\mathcal{B}(s_1,\ldots, s_n)$ concretely.

$\mathcal{B}^{\mathrm{tr}}(s_1,\ldots, s_n)$ is the 
left coideal subalgebra of $U_p(\mathfrak{g}^{\mathrm{tr}})$ 
generated by 
\begin{align}
\hat{b}_i &= \frac{s_{i+1}}{s_i}f_i 
+\frac{s_i}{s_{i+1}}p^{2} k^{-1}_ie_i 
+ \frac{\I \epsilon p}{1-p^2}k^{-1}_i\quad (i \in \Z_n).
\label{bais}
\end{align}
Of course it actually depends on the ratio of $s_i$'s, 
hence forms an $(n-1)$-parameter family.

$\mathcal{B}^{r,r'}_{k,k'}(s_1,\ldots, s_n)$ 
with $r\le k$ and $r'\le k'$ is the 
left coideal subalgebra of $U_p(\mathfrak{g}^{r,r'})$ 
generated by 
\begin{alignat}{2}
\hat{b}_0 &= f_0 + \frac{s_1'}{s_1}p k_0^{-1}e_0 
&\quad (r,k)&=(1,2),
\\
&=  f_0 + \frac{s_1'}{s_1}p k_0^{-1}e_0 
+ \frac{\epsilon+\I p}{s_1(1-p^2)}k^{-1}_0
&\quad (r,k)&=(1,1),
\\
&=  f_0 + \Bigl(\frac{s_2'}{s_1}p\Bigr)^2k_0^{-1}e_0 
+ \Bigl(\frac{s_2'}{s_1}\Bigr)\frac{\I \epsilon p}{1-p^2}k^{-1}_0
&\quad (r,k)&=(2,2),
\\
\hat{b}_i &= \ \frac{s_{i+1}}{s_i}f_i 
+\frac{s_i}{s_{i+1}}p^{2} k^{-1}_ie_i 
+ \frac{\I \epsilon p}{1-p^2}k^{-1}_i
\quad(0 < i <n),
\label{bg}\\
\hat{b}_n &= f_n + \frac{s_n}{s'_n}p k_n^{-1}e_n
&\quad (r',k')&=(1,2),
\\
&=  f_n + \frac{s_n}{s'_n}p k_n^{-1}e_n 
+ \frac{\epsilon+\I p}{s'_n(1-p^2)}k^{-1}_n
&\quad (r',k')&=(1,1),
\\
&=  f_n + \Bigl(\frac{s_{n-1}}{s'_n}p\Bigr)^2k_n^{-1}e_n 
+ \Bigl(\frac{s_{n-1}}{s'_n}\Bigr)\frac{\I \epsilon p}{1-p^2}k^{-1}_n
&\quad (r',k')&=(2,2),
\label{bns}
\end{alignat}
where $s'_j = q/s_j = -\I \epsilon/( ps_j)$ in all the cases.
The definitions (\ref{bais})--(\ref{bns}) 
reduce to (\ref{bai})--(\ref{bn}) in the
homogeneous case $s_1=\cdots =s_n=s=q/s'$.

The argument so far is summarized as 
a multi-parameter generalization of 
Proposition \ref{pr:kpai}:

\begin{proposition}\label{pr:ks}
Under the condition $s_js'_j = q = -\I \epsilon p^{-1}\,(1 \le j \le n)$,
the $K$ matrices defined by the inhomogeneous matrix products 
(\ref{iktr}) and (\ref{ikkk}) satisfy the intertwining relation
\begin{align}\label{my}
K(z,q; s_1,\ldots, s_n) \pi_z(b) = 
\pi^\ast_{z^{-1}}(b)K(z,q; s_1,\ldots, s_n) \qquad 
(b \in \mathcal{B}(s_1,\ldots, s_n)),
\end{align}
where
$(K,\mathcal{B}) 
= (K^{\mathrm{tr}}, \mathcal{B}^{\mathrm{tr}})$
for $U_p(\mathfrak{g}^{\mathrm{tr}})$ and 
$(K,\mathcal{B}) = (K^{k,k'}, \mathcal{B}^{r,r'}_{k,k'})$
for $U_p(\mathfrak{g}^{r,r'})$ with $r\le k$ and $r' \le k'$.
\end{proposition}
One can check the claim directly as 
the proof of Proposition \ref{pr:kpai}.

\subsection{Relation to generalized Onsager algebras}\label{ss:o}

Consider the element 
\begin{align}\label{A}
A_i = c_i f_i + \bar{c}_i k^{-1}_i e_i + d_i k^{-1}_i  \in U_p
\quad (0 \le i \le n'), 
\end{align}
where $c_i, \bar{c}_i, d_i$ are coefficients.
In view of 
\begin{align}
\Delta A_i = k_i^{-1}\otimes A_i +
(c_i f_i + \bar{c}_ik^{-1}_i e_i) \otimes 1, 
\end{align}
the subalgebra $\mathcal{A}\subset U_p$ 
generated by $A_0,\ldots, A_n$ 
forms a left coideal 
$\Delta \mathcal{A} \subset U_p \otimes \mathcal{A}$.

In \cite{BB} the condition on the coefficients $c_i, \bar{c}_i, d_i$ are determined
so that there are a closed set of relations among $A_i$'s that 
make $\mathcal{A}\subset U_p$ isomorphic 
to the {\em generalized $p$-Onsager algebra}.
Let us concentrate on 
$U_p(\mathfrak{g}^{\mathrm{tr}})$ and $U_p(\mathfrak{g}^{r,r'})$
treated in this paper. 
Then except for $U_p(A^{(1)}_1)$ which is the smallest example 
from the former,
the condition is described as follows:
\begin{alignat}{2}
&d_i\bigl( c_j\bar{c}_j + d_j^2(1-p^2)^2\bigr) = 
d_j\bigl( c_i\bar{c}_i + d_i^2(1-p^2)^2\bigr) = 0
&\quad &\text{if}\;\,(a_{ij},a_{ji})=(-1,-1),
\label{o1}\\
&d_j\bigl( c_i\bar{c}_i + d_i^2(1-p^2)^2\bigr) = 0
&\quad &\text{if}\;\,(a_{ij},a_{ji})=(-1,-2),
\label{o2}
\end{alignat}
where $(a_{ij})_{0 \le i,j \le n'}$ is the Cartan matrix appearing in (\ref{uqdef}).
(The case $U_p(A^{(1)}_1)$ is exceptional in that 
there is no constraint on the parameters.)
In the resulting generalized $p$-Onsager algebra,  
the structure constants themselves are dependent on  
$c_i, \bar{c}_i$ obeying the above constraint.
We call such a coideal subalgebra of $U_p$ 
an {\em Onsager coideal} for simplicity. 

It turns out that 
$\mathcal{B}^{\mathrm{tr}}(s_1,\ldots, s_n)$ and 
$\mathcal{B}^{r,r'}_{k,k'}(s_1,\ldots, s_n)$ defined via  
the generators $\hat{b}_i$'s in 
(\ref{bais})--(\ref{bns}) satisfy (\ref{o1}) and (\ref{o2}).
This fact can be checked case by case.
In fact for the generic situation where the adjacent nodes $i,j$ are
apart from boundaries of a Dynkin diagram, 
(\ref{o1}) follows from (\ref{bais}) or (\ref{bg}) 
by setting $c_i = \frac{s_{i+1}}{s_i},
\bar{c}_i = \frac{s_i}{s_{i+1}}$ and 
$d_i = \frac{\I \epsilon}{1-p^2}$.
As an example of non generic cases,
consider $U_p(D^{(2)}_{n+1})$  
for which (\ref{o2}) matters
for $(i,j)=(1,0)$ and $(n-1,n)$.
It indeed holds because
$\hat{b}_i$ in either case is given by (\ref{bg})
hence $c_i\bar{c}_i + d_i^2(1-p^2)^2=0$ is valid.
By such verification one can establish 
\begin{proposition}\label{pr:on}
$\mathcal{B}^{\mathrm{tr}}(s_1,\ldots, s_n)$ 
and 
$\mathcal{B}^{r,r'}_{k,k'}(s_1,\ldots, s_n)\, (r \le k, r' \le k')$
are Onsager coideals. 
\end{proposition}

Synthesizing Proposition \ref{pr:ks} and Proposition \ref{pr:on}
we arrive at the conclusion of this section: 
\begin{theorem}\label{th:pp}
For $(K, R, \mathcal{B}) 
= (K^{\mathrm{tr}},  R^{\mathrm{tr}},  \mathcal{B}^{\mathrm{tr}})$  
and  $(K^{k,k'}, R^{r,r'}, \mathcal{B}^{k,k'})$ with $r \le k, r' \le k'$,
the matrix product construction (\ref{iktr})--(\ref{ikkk}) 
leads to the intertwiner $K(z,q; s_1,\ldots, s_n)$ of the 
Onsager coideal $\mathcal{B}(s_1,\ldots, s_n)$
which satisfies the reflection equation
\begin{equation}\label{mtkn}
\begin{split}
&K_2(y,q; s_1,\ldots, s_n) R^\ast_{21}(xy,q)
K_1(x,q; s_1,\ldots, s_n) R_{12}(x/y,q)
\\
&=
R^{\ast\ast}_{21}(x/y,q)K_1(x,q; s_1,\ldots, s_n)
R^\ast_{12}(xy,q)K_2(y,q; s_1,\ldots, s_n).
\end{split}
\end{equation}
\end{theorem}

We note that the constraint like (\ref{o1}) and (\ref{o2}) 
has also emerged in \cite[Rem.2]{KOY} as a necessary condition for the 
existence of the intertwining $K$ matrix for the 
symmetric tensor representations of $U_p(A^{(1)}_{n-1})$.

\section{Summary}\label{s:cdk}
We have introduced the Onsager coideals 
$\mathcal{B}(s_1,\ldots, s_n)$ with generators in (\ref{bais})--(\ref{bns})
and determined their intertwiner, the $K$ matrices, characterized by (\ref{my}) for 
the fundamental representations of $U_p(A^{(1)}_{n-1})$ and the 
spin representations of $U_p(B^{(1)}_n)$, $U_p(D^{(1)}_n)$ and $U_p(D^{(2)}_{n+1})$
specified in (\ref{rep}).
The results are expressed in the matrix product form as in 
(\ref{iktr}) and (\ref{ikkk}), and 
satisfy the reflection equation (\ref{mtkn}).
(The companion $R$ matrices are also given by the matrix product formulas 
in (\ref{Rtr})--(\ref{r1-}).)

This paper elucidates a prominent role of Onsager coideals 
in the study of reflection $K$ matrices.
It also sheds new light 
on the entirely different approach in \cite{KP} based on the 
three dimensional integrability and the 
quantized reflection equation digested in Appendix \ref{app:qre}. 
In fact the conjectural results there for 
$U_p(B^{(1)}_n)$, $U_p(D^{(1)}_n)$, $U_p(D^{(2)}_{n+1})$
have been established here for the first time.
We have stressed two essential aspects; 
connection to the Onsager coideals and 
matrix product construction by $q$-Bosons.
Their interrelation deserves a further investigation.

\newpage

\appendix

\section{Quantized reflection equation}\label{app:qre}
Let us recall the basic ingredients that led to the 
matrix product solutions (\ref{ktr}), (\ref{kkk}) to the reflection equation.
We regard $L$ (\ref{L}) as an operator 
in $\in \mathrm{End}(\C^2 \otimes \C^2 \otimes F_{q^2})$ 
defined by 
\begin{align}
L(u_i \otimes u_j \otimes |m\rangle) = 
\sum_{a,b \in \{0,1\}}u_a \otimes u_b \otimes L^{a,b}_{i,j}|m\rangle,
\end{align}
where $u_0, u_1$ are bases of $\C^2 = \C u_0 \oplus \C u_1$.
We also introduce its variants
corresponding to (\ref{R1})--(\ref{Rss1}): 
\begin{align}
L^\ast(u_i \otimes u_j \otimes |m\rangle) &= 
\sum_{a,b \in \{0,1\}}u_a \otimes u_b \otimes L^{a,b \,\ast}_{i,j}|m\rangle,
\quad\;
L^{a,b\, \ast}_{i,j} = L^{a,1-b}_{i,1-j},
\label{Ls}\\
L^{\ast\ast}(u_i \otimes u_j \otimes |m\rangle) &= 
\sum_{a,b \in \{0,1\}}
u_a \otimes u_b \otimes L^{a,b\, \ast\ast}_{i,j}|m\rangle,
\quad L^{a,b\, \ast\ast}_{i,j} = L^{1-a,1-b}_{1-i,1-j}.
\label{Lss}
\end{align}
Similarly we regard $G$ (\ref{G}) as an operator in
$\mathrm{End}(\C^2 \otimes F_q)$ defined by
\begin{align}
G(u_i \otimes |m\rangle ) = \sum_{a \in \{0,1\}}u_a 
\otimes G_i^a |m\rangle.
\end{align}

The quantized reflection equation \cite{KP} is the 
reflection equation up to conjugation:
\begin{align}\label{qre}
(L^{\ast\ast}_{21}G_1L^\ast_{12}G_2) \mathcal{K}
= \mathcal{K}(G_2L^{\ast}_{21}G_1L_{12}),
\end{align}
where the conjugation operator $\mathcal{K} \in \mathrm{End}
(F_{q^2}\otimes F_q \otimes F_{q^2} \otimes F_q)$ is called 3D $K$.
It is depicted as follows:
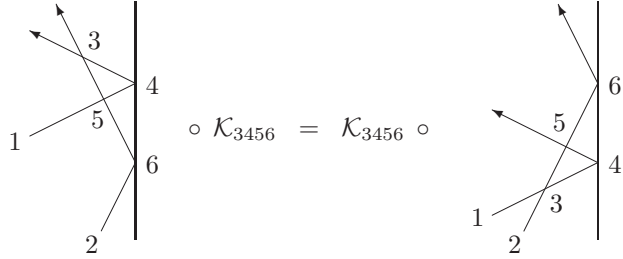
\begin{align*}
\begin{picture}(200,94)(-20,-8)
\put(0,-9){\line(0,1){90}}
\put(0,20){\line(-1,-2){13}}\put(0,20){\vector(-1,2){30}}
\put(0,50){\line(-2,-1){40}}\put(0,50){\vector(-2,1){40}}
\put(-18,63){3}\put(4,46){4}\put(-16.7,33){5}\put(4,16){6}
\put(-48,25){1} \put(-19,-14){2}
\put(20,30){$\circ\; \,\mathcal{K}_{3456}\;\;=\;\; \mathcal{K}_{3456}\,\;\circ$}
\put(175,0){
\put(0,-9){\line(0,1){90}}
\put(0,50){\line(-1,-2){28}}\put(0,50){\vector(-1,2){15}}
\put(0,20){\line(-2,-1){40}}\put(0,20){\vector(-2,1){40}}
\put(-18.3,1.2){3}\put(4,16){4}\put(-17,32){5}\put(4,46){6}
\put(-48,-4){1} \put(-34,-14){2}
}
\end{picture}
\end{align*}
The equality (\ref{qre}) is to hold in 
$\mathrm{End}(\overset{1}{\C^2}
\otimes \overset{2}{\C^2} \otimes 
\overset{3}{F_{q^2}}\otimes 
\overset{4}{F_{q}}\otimes 
\overset{5}{F_{q^2}}\otimes 
\overset{6}{F_{q}})$, 
where the labels $1,\ldots, 6$ are assigned just for explanation.
The vertical straight line on each side is the reflecting boundary.
The other arrows labeled with $1,2$ signify $\C^2$.
At each vertex one should imagine an extra arrow penetrating it perpendicularly 
to the sheet from the back to the front.
Those going through $4,6$ (resp. $3,5$) are assigned with $F_{q}$ 
(resp. $F_{q^2}$). 
One may regard $\mathcal{K}_{3456}$ in the left (resp. right) hand side
as a point in the back (resp. front) of the diagram where 
the four arrows going toward (resp. coming from) 
the vertices $3,4,5,6$  intersect.
If the labels of these spaces are exhibited, (\ref{qre}) becomes
\begin{align}
L^{\ast\ast}_{213}G_{14}L^\ast_{125}G_{26}\mathcal{K}_{3456}
= \mathcal{K}_{3456}G_{26}L^\ast_{215}G_{14}L_{123}.
\end{align}
In terms of the operators in (\ref{L}), (\ref{G}), (\ref{Ls}) and (\ref{Lss}),
this is expressed as
\begin{align}\label{qre2}
\Bigl(\;\sum 
L^{j_3,i_3\,\ast\ast}_{j_2 i_2} \otimes 
G^{i_2}_{i_1} \otimes
L^{i_1 j_2 \ast}_{i_0 j_1} \otimes
G^{j_1}_{j_0}\Bigr) \mathcal{K}
= \mathcal{K}\Bigl(\;\sum
L^{i_1j_1}_{i_0 j_0} \otimes 
G^{i_2}_{i_1} \otimes 
L^{j_2 i_3\, \ast}_{j_1 i_2} \otimes 
G^{j_3}_{j_2}
\Bigr)
\end{align}
for any $i_0,j_0, i_3, j_3 \in \{0,1\}$,
where the sums are taken over $i_1, i_2, j_1, j_2 \in \{0,1\}$.
The explicit form of the sixteen equations ({\ref{qre2}) is 
available in \cite[App. A]{KO1}\footnote{The operators 
$\Apm, \OK, \apm, \ok$ in \cite{KO1} has the same meaning as 
the ones in this paper, although the 3D $K$ is denoted by 
$\mathscr{K}$ rather than $\mathcal{K}$ here.}.
We remark that (\ref{qre2}) becomes actually {\em independent} of the parameters 
$\epsilon, s, s'$ entering $L$ (\ref{L}) and $G$ (\ref{G}) as long as the 
constraint (\ref{qp}) is satisfied.
The 3D $K$ $\mathcal{K}$ is uniquely determined from the 
quantized reflection equation 
up to normalization as in \cite[eqs.(3.27)-(3.28)]{KO1}.
Then the matrix product solutions (\ref{ktr})--(\ref{k22}) to the 
ordinary reflection equation are constructed following the 
reduction method explained in \cite[Sec.5]{KP}
assuming \cite[eq.(78)]{KP} for $U_p(\mathfrak{g}^{r,r'})$.

\section{Generators of $\mathcal{B}^{r,r'}_{k,k'}$}\label{app:B}

Let us write down the generators 
$b_0, \ldots, b_n$ of $\mathcal{B}^{r,r'}_{k,k'}$ 
specified in (\ref{b0})--(\ref{drk}) explicitly.
For any $\mathcal{B}^{r,r'}_{k,k'}$ with 
$r\le k$ and  $r'\le k'$ we have 
\begin{align}
b_i 
= f_i +p^{2}k^{-1}_ie_i 
+ \frac{\I \epsilon p}{1-p^2}k^{-1}_i
\quad (0 < i <n).
\end{align}
In what follows we list $b_0, b_n$ only.
\begin{align}
\mathcal{B}^{1,1}_{1,1}:\qquad 
b_0 
&= f_0 + \frac{s'}{s}p^{}k_0^{-1}e_0
+ 
\frac{\epsilon + \I p}{s(1-p^2)}k^{-1}_0,
\\
b_n 
&= f_n + \frac{s}{s'}p^{} k^{-1}_ne_n 
+
\frac{\epsilon +\I p}{s'(1-p^2)}k^{-1}_n,
\label{b11}
\\
\mathcal{B}^{2,1}_{2,1}:\qquad 
b_0 
&= f_0 + \frac{s'^2}{s^2}p^{2}k_0^{-1}e_0
+ 
\frac{1}{s^2(1-p^2)}k^{-1}_0,
\\
b_n 
&= f_n + \frac{s}{s'}p^{}  k^{-1}_ne_n
+
\frac{\epsilon +\I p}{s'(1-p^2)}k^{-1}_n,
\label{b21}
\\
\mathcal{B}^{1,1}_{2,1}:\qquad
b_0 &= f_0 +  \frac{s'}{s}p^{}k_0^{-1}e_0,
\\
b_n &=  f_n + \frac{s}{s'}p^{} k^{-1}_n e_n
+\frac{\epsilon +\I p}{s'(1-p^2)}k^{-1}_n, 
\\
\mathcal{B}^{1,2}_{1,2}:\qquad 
b_0 
&= f_0 + \frac{s'}{s}p^{}k_0^{-1}e_0
+ 
\frac{\epsilon + \I p}{s(1-p^2)}k^{-1}_0,
\\
b_n 
&= f_n + \frac{s^2}{s'^2}p^{2} k^{-1}_ne_n 
+
\frac{1}{s'^2(1-p^2)}k^{-1}_n,
\\
\mathcal{B}^{1,1}_{1,2}:\qquad 
b_0 
&= f_0 + \frac{s'}{s}p^{}k_0^{-1}e_0
+ 
\frac{\epsilon + \I p}{s(1-p^2)}k^{-1}_0,
\\
b_n &= f_n + \frac{s}{s'}p^{}  k^{-1}_ne_n,
\\
\mathcal{B}^{2,2}_{2,2}:\qquad  
b_0 
&= f_0 + \frac{s'^2}{s^2}p^{2}k_0^{-1}e_0
+ 
\frac{1}{s^2(1-p^2)}k^{-1}_0,
\\
b_n 
&= f_n + \frac{s^2}{s'^2}p^{2}  k^{-1}_ne_n
+
\frac{1}{s'^2(1-p^2)}k^{-1}_n,
\\
\mathcal{B}^{2,1}_{2,2}:\qquad
b_0 
&= f_0 + \frac{s'^2}{s^2}p^{2}k_0^{-1}e_0
+ 
\frac{1}{s^2(1-p^2)}k^{-1}_0,
\\
b_n &= f_n + \frac{s}{s'}p^{} k^{-1}_ne_n ,
\\
\mathcal{B}^{1,2}_{2,2}:\qquad
b_0 &= f_0 +  \frac{s'}{s}p^{}k_0^{-1}e_0,
\\
b_n 
&= f_n + \frac{s^2}{s'^2}p^{2} k^{-1}_ne_n 
+
\frac{1}{s'^2(1-p^2)}k^{-1}_n,
\\
\mathcal{B}^{1,1}_{2,2}:\qquad
b_0 &= f_0 +  \frac{s'}{s}p^{}k_0^{-1}e_0,
\\
b_n &= f_n + \frac{s}{s'}p^{}  k^{-1}_ne_n.
\end{align}

\section*{Acknowledgments}
A.K. thanks Vladimir Mangazeev for sending the preprint \cite{MLu}.
He is supported by Grants-in-Aid for Scientific Research 
No.~18H01141 from JSPS.
M.O. is supported by Grants-in-Aid for Scientific Research No.~15K13429
and No.~16H03922 from JSPS.


\begin{thebibliography}{99}

\bibitem{AM}
M.~F.~Atiyah and I.~G.~MacDonald,
{\it Introduction to Commutative Algebra}, 
Reading, MA: Addison-Wesley (1969).

\bibitem{BB}
P.~Baseilhac, S.~Belliard,
Generalized $q$-Onsager algebras and boundary affine Toda field theories,
Lett. Math. Phys. {\bf 93} 213--228 (2010). 

\bibitem{BFKZ} M.~T.~Batchelor, V.~Fridkin, A.~Kuniba, Y.~K.~Zhou, 
Solutions of the reflection equation for face 
and vertex models associated with 
$A^{(1)}_n$,
$B^{(1)}_n$,
$C^{(1)}_n$,
$D^{(1)}_n$ and 
$A^{(2)}_n$,
Phys. Lett. B{\bf 376} 266--274 (1996).

\bibitem{Bax}R.~J.~Baxter,
\textit{Exactly solved models in statistical mechanics},
Academic Press (1982).

\bibitem{BS}
V.~V.~Bazhanov and S.~M.~Sergeev, 
Zamolodchikov's tetrahedron equation 
and hidden structure of quantum groups, 
J. Phys. A: Math. Theor. {\bf 39}  3295--3310  (2006).

\bibitem{Ch}I.~V.~Cherednik,
Factorizing particles on a half-line and root systems,
Theor. Math. Phys. {\bf 61}  35--44 (1984).

\bibitem{DO}E.~Date and M.~Okado,
Calculation of excitation spectra of the spin model related the
vector representaton of the quantized affine algebra of type $A^{(1)}_n$,
Int. J. Modern. Phys. A {\bf 9} 399--417 (1994). 

\bibitem{DM}
G.~Delius and N.~J.~MacKay,
Quantum group symmetry in sine-Gordon and affine Toda field theories on the half-line,
Commun. Math. Phys. {\bf 233}  173--190 (2003).

\bibitem{D}
V.~G.~Drinfeld,
Quantum groups. 
Proceedings of the ICM, Vols. 1, 2
(Berkeley, Calif., 1986), pp. 798--820.  
Am. Math. Soc., Providence (1987).

\bibitem{G}
G.~M.~Gandenberger,
New non-diagonal solutions to the $a^{(1)}_n$ boundary Yang-Baxter equation,
arXiv:hep-th/9911178.

\bibitem{Ji}
M.~Jimbo,
A $q$-difference analogue of $U(g)$ and the Yang-Baxter equation,
Lett. Math. Phys. {\bf 10} 63--69 (1985).

\bibitem{Kac}
V.~G.~Kac,
{\it Infinite dimensional Lie algebras}, third ed.,
Cambridge University Press (1990).

\bibitem{Ko} S.~Kolb, 
Quantum symmetric Kac-Moody pairs. Adv. Math. {\bf 267}  395--469 (2014).

\bibitem{Kul} P.~P.~Kulish,
Yang-Baxter equation and reflection equations in integrable models,
in {\em Low-dimensional models in statistical physics and quantum field theory} 
(Schladming, 1995), Lect. Note. Phys. {\bf 469} 125--144.

\bibitem{KO1} 
A.~Kuniba and M.~Okado,
Tetrahedron and 3D reflection equations
from quantized algebra of functions,
J. Phys. A: Math.Theor. {\bf 45}  465206 (27pp) (2012).

\bibitem{KOY} A.~Kuniba, M.~Okado, A.~Yoneyama,
Matrix product solution to the reflection equation 
associated with a coideal subalgebra of $U_q(A^{(1)}_{n-1})$,
arXiv:1812.03767, Lett. Math. Phys. in press.

\bibitem{KP}
A.~Kuniba and V.~Pasquier,
Matrix product solutions 
to the reflection equation from three dimensional integrability,
J. Phys. A: Math. Theor. {\bf 51} 255204 (26pp)  (2018).

\bibitem{KS} 
A.~Kuniba and S.~Sergeev,
Tetrahedron equation and quantum $R$ matrices for spin
representations of  
$B^{(1)}_n$, $D^{(1)}_n$ and $D^{(2)}_{n+1}$,
Commun. Math. Phys. {\bf 324} 695--713 (2013).

\bibitem{ML} R.~Malara and A.~Lima-Santos, 
On 
$A^{(1)}_{n-1}$,
$B^{(1)}_{n}$,
$C^{(1)}_{n}$,
$D^{(1)}_{n}$,
$A^{(2)}_{2n}$,
$A^{(2)}_{2n-1}$ and $D^{(2)}_{n+1}$ reflection $K$-matrices,
J. Stat. Mech. {\bf 0609} P09013 (2006).

\bibitem{MLu}V.~V.~Mangazeev and  X.~Lu,
Boundary matrices for the higher spin six vertex model,
arXiv:1903.00274.

\bibitem{NR}R.~I.~Nepomechie and A.~L.~Retore,
Surveying the quantum group symmetries of integrable open spin chains,
arXiv:1802.04864.

\bibitem{O}M.~Okado,
Quantum $R$ matrices related to the spin representations of
$B_n$ and $D_n$,
Commun. Math. Phys. {\bf 134} 467--486 (1990).

\bibitem{On}L.~Onsager,
Crystal statistics. I. A two-dimensional model with an order-disorder transition,
Phys. Rev. {\bf 65} 117--149  (1944). 

\bibitem{RV}
V.~Regelskis and B.~Vlaar,
Reflection matrices, coideal subalgebras 
and generalized Satake diagrams of affine type,
arXiv:1602.08471.

\bibitem{Sk} E.~K.~Sklyanin,
Boundary conditions for integrable quantum systems,
J. Phys. A: Math. Gen. {\bf 21} 2375--2389 (1988).

\bibitem{T} P.~Terwilliger,
An algebraic approach to the Askey scheme of orthogonal polynomials,
in {\it Orthogonal Polynomials and Special Functions},  
Lect. Note. in Math. {\bf 1883}
Springer 225--330 (2006). 

\end{thebibliography}
\end{document}